\documentclass[final,5p,times,twocolumn,nopreprintline]{elsarticle}

\usepackage{amsmath,amssymb,amsthm,mathtools}
\usepackage{bm}
\usepackage{graphicx}
\usepackage{color}
\usepackage{url}
\usepackage{hyperref}
\usepackage{enumitem}
\usepackage{booktabs}
\usepackage{empheq}
\usepackage{placeins}

\biboptions{numbers,sort&compress}
\mathtoolsset{showonlyrefs}
\hypersetup{
    colorlinks=true,
    linkcolor=blue,
    citecolor=blue,
    urlcolor=blue
}

\theoremstyle{plain}
\newtheorem{theorem}{Theorem}[section]
\newtheorem{proposition}[theorem]{Proposition}

\newtheorem{corollary}[theorem]{Corollary}

\theoremstyle{definition}
\newtheorem{definition}[theorem]{Definition}

\theoremstyle{remark}
\newtheorem{remark}[theorem]{Remark}

\numberwithin{equation}{section}

\date{}
\begin{document}

\begin{frontmatter}

\title{
Bifurcation structure and mesa pattern formation in a one-component nonlocal adhesion model with population pressure and degenerate mobility
}

\author[aff1]{Shimpei Makida}
\ead{makida@math.ryukoku.ac.jp}

\author[aff1]{Hideki Murakawa}
\ead{murakawa@math.ryukoku.ac.jp}

\affiliation[aff1]{organization={Faculty of Advanced Science and Technology, Ryukoku University},
                   addressline={1-5 Yokotani, Seta Oe-cho},
                   city={Otsu},
                   postcode={520-2194},
                   state={Shiga},
                   country={Japan}}

\begin{abstract}
We analyze pattern formation from a homogeneous steady state in a one-component nonlocal adhesion model with population pressure and degenerate mobility. First, using linear stability analysis, we derive the instability threshold and a selection rule for the fastest-growing mode, and elucidate the mechanism by which the selected wavenumber shifts toward lower wavenumbers as the mean density increases. We then perform a weakly nonlinear analysis near the critical adhesion strength and derive an explicit expression for the Landau coefficient in the Stuart--Landau equation. This expression shows that the critical bifurcation is classified as supercritical or subcritical according to the mean density, the nonlinear exponent, and the second-harmonic response ratio of the kernel. Furthermore, we show that a large nonlinear exponent promotes a transition to subcriticality and confirm, through numerical bifurcation analysis and time-dependent simulations, a bifurcation structure with a fold point and the formation of mesa patterns. Finally, through the energy limit as \(m\to\infty\), we relate the observed mesa profiles to a capacity-constrained limiting structure.
\end{abstract}


\end{frontmatter}

% ============================================================

\section{Introduction}
\label{sec:introduction}

How cell populations aggregate, segregate, and form spatially organized structures is a fundamental question in developmental biology, tissue formation, wound healing, and tumor growth. During these processes, cells do not move purely at random; rather, they exhibit directed motion in response to interactions with neighboring cells and the extracellular environment. In particular, cell-cell adhesion is a fundamental mechanism that draws cells together and drives aggregation, boundary formation, alignment, and self-organization. For example, N-cadherin has been reported to play an essential role in neuronal self-organization in the Drosophila visual system, indicating that cell-cell adhesion is not restricted to microscopic cell-cell contacts but also contributes to macroscopic pattern formation at the tissue scale \cite{TrushLiuHanNakaiTakayamaMurakawaCarrilloTakechiHakedaSuzukiSuzukiSato2019}.

When collective motion driven by cell-cell adhesion is described within a continuum framework, it is natural to incorporate the nonlocal nature of cell sensing, whereby cells detect other cells within a finite sensing range and move in response to their spatial distribution.
From this perspective, beginning with the nonlocal integro-differential adhesion model introduced by Armstrong--Painter--Sherratt~\cite{ArmstrongPainterSherratt}, extensive studies have investigated the mathematical structure and pattern formation of cell adhesion models involving nonlocal advection \cite{ButtenschoenHillen2021,MurakawaTogashi2015,PainterHillenPotts2024}.

In real cell populations, however, density cannot increase without bound. Cell volume and excluded-volume effects impose a natural upper bound on the density. In models incorporating only a simple nonlocal attraction, the cell density may become excessively concentrated, potentially producing nonphysical singularities beyond the biologically admissible density range. In this study, we represent this high-density effect by the population pressure
$
    p(u)=u^m.
$
The pressure gradient produces local dispersal away from high-density regions, whereas the adhesion mobility \(1-u^m\) decreases near the capacity limit, thereby suppressing local advection driven by the nonlocal adhesion field. By incorporating these pressure effects and density saturation, the model can describe pattern formation driven by nonlocal adhesion while preventing excessive aggregation \cite{AnguigeSchmeiser2009,CarrilloMurakawaSatoTogashiTrush2019,CarrilloChenWangWangZhang2020}.

The one-component nonlocal adhesion model considered in this paper is given for the density \(u=u(x,t)\) by
\begin{equation}
    \frac{\partial  u}{\partial t}
    =
    \nabla\cdot(u\nabla u^m)
    -
    a\nabla\cdot
    \left(
        u(1-u^m)\nabla(K*u)
    \right)
    \label{eq:main_model}
\end{equation}
Here, \(m\ge1\) is an exponent characterizing the nonlinearity of the population pressure, \(a>0\) is the adhesion strength, and \(K\) is a nonlocal kernel representing cell-cell adhesive interactions within the sensing range. The first term represents local dispersal generated by the gradient of the population pressure \(p(u)=u^m\). This term can be written as
$\frac{m}{m+1}\Delta u^{m+1}$
and is therefore a porous-medium-type nonlinear degenerate diffusion term. The second term represents nonlocal advection induced by cell-cell adhesion. In particular, the factor \(1-u^m=1-p(u)\) in the adhesive flux becomes small as the population pressure approaches the capacity limit, thereby suppressing adhesion-driven advection.

Ishii--Murakawa--Tanaka \cite{IshiiMurakawaTanaka2026} proved that, when the initial data satisfy \(0\le u_0\le1\), a corresponding weak solution exists and the bound \(0\le u\le1\) is preserved throughout the evolution. Building on this well-posedness result, we focus on the destabilization of a homogeneous steady state, the bifurcation structure near the critical point, and mesa pattern formation in the large-amplitude regime.

The first objective of this study is to clarify the linear stability and mode selection of the homogeneous steady state. For a homogeneous state with mean density \(\bar u\in(0,1)\), we derive the dispersion relation and show that the instability threshold is determined by the adhesion strength, the mean density, and the Fourier response of the kernel. We also introduce an effective adhesion ratio that quantifies the strength of nonlocal adhesion relative to local diffusion, and explain the mechanism by which the wavenumber with the maximum growth rate shifts toward lower wavenumbers as the mean density increases. On a finite periodic domain, this shift toward lower wavenumbers manifests itself as a stepwise decrease in the selected Fourier mode number.

The second objective is to determine whether the critical bifurcation from the homogeneous state is supercritical or subcritical. On a one-dimensional periodic domain where a simple critical wavenumber \(k_c\) is selected, we derive an amplitude equation for the critical mode using a multiple-scale expansion. We refer to the cubic coefficient in the resulting Stuart--Landau equation as the Landau coefficient \(c_{\mathrm L}\), whose sign determines the type of the critical bifurcation. We express \(c_{\mathrm L}\) explicitly as a function of the mean density \(\bar u\), the nonlinear exponent \(m\), and the second-harmonic response ratio \(\kappa=\widehat K(2k_c)/\widehat K(k_c)\). This representation provides a unified framework for describing the effects of the kernel shape and the system size on the bifurcation type through the second-harmonic response ratio.

This analysis yields an important conclusion regarding the adhesion strength \(a\). The adhesion strength determines the critical value \(a_c\) at which the homogeneous state becomes linearly unstable. By contrast, whether the bifurcation arising near the critical point is supercritical or subcritical is determined by the sign of the Landau coefficient after \(a_c\) has been eliminated using the criticality condition. This sign depends primarily on \(\bar u\), \(m\), and \(\kappa\). Thus, the adhesion strength determines the onset of instability, whereas the bifurcation type itself is governed by the volume-filling nonlinearity and the harmonic response of the kernel.

The third objective is to explain mesa pattern formation for large \(m\) through the connection between the weakly nonlinear analysis and the energy limit as \(m\to\infty\). From the explicit expression for the Landau coefficient, we show that, when the nonlinear exponent \(m\) is sufficiently large, the critical bifurcation tends to become subcritical over a broad range of kernel responses. This implies that small-amplitude wave-like patterns are more likely to transition to large-amplitude aggregates than to saturate smoothly. Numerical simulations confirm that, as \(m\) increases, the density develops a flat region near the capacity limit \(u=1\) and approaches a mesa profile with sharp interfaces.

The mesa limit for nonlinear nonlocal aggregation--diffusion equations has been analyzed in detail by Carrillo--Gvalani \cite{CarrilloGvalani2021}. We apply their mesa-limit framework to the volume-filling entropy of the present model.
Specifically, we show that, as \(m\to\infty\), the local entropy term vanishes under the capacity constraint \(0\le u\le1\), and that the energy functionals converge uniformly and \(\Gamma\)-converge to a capacity-constrained nonlocal attraction energy.
We further show that, when the Fourier coefficients of the kernel are nonnegative, the limiting problem admits a characteristic-function-type minimizer. This result links the mesa profiles observed for large \(m\), consisting of saturated regions with \(u\simeq1\) and background regions with \(u\simeq0\), to the limiting energy structure as \(m\to\infty\).

The remainder of this paper is organized as follows. Section~\ref{sec:model} introduces the model and its basic structure. Section~\ref{sec:linear_stability} addresses the linear stability and mode selection of the homogeneous steady state. Section~\ref{sec:weakly_nonlinear} derives the Stuart--Landau equation through a weakly nonlinear analysis, provides an explicit expression for the Landau coefficient, and analyzes the transition to subcriticality induced by large \(m\). Section~\ref{sec:mesa_limit} examines the mesa limit as \(m\to\infty\) and the capacity-constrained energy structure, and confirms mesa pattern formation through numerical simulations. Finally, Section~\ref{sec:conclusion} summarizes the results and discusses directions for future research. The Fourier responses of the kernels, the numerical methods, the construction of the steady-state bifurcation diagrams, and details of some proofs are presented in the appendices.

\section{Model and basic structure}
\label{sec:model}

In this section, we consider the model \eqref{eq:main_model} for the cell density \(u=u(x,t)\) on the \(N\)-dimensional domain \(\Omega=[-L,L)^N\) under periodic boundary conditions. We summarize the kernels used in this study, known results on the existence and boundedness of solutions, and the formal energy structure of the model.

\subsection{Kernel}
\label{subsec:model_formulation}

The nonlocal interaction is given by convolution with a periodized kernel \(K\):
\begin{equation}
    (K*u)(x,t)
    =
    \int_\Omega K(x-y)u(y,t)\,dy
    \label{eq:convolution}
\end{equation}
To define the nonlocal interaction under periodic boundary conditions, we periodize a radially symmetric base kernel \(K_0(x)\), defined on the whole space \(\mathbb R^N\), with period \(2L\) (cf. \cite{IshiiMurakawaTanaka2026}). Specifically, we define
\begin{equation}
    K(x)
    =
    \sum_{\ell\in\mathbb Z^N}
    K_0(x-2L\ell),
    \qquad
    x\in\Omega.
    \label{eq:periodized_kernel}
\end{equation}
Then \(K\) is periodic on \(\Omega\), and the convolution \eqref{eq:convolution} is consistent with the periodic boundary conditions.

Moreover, for each admissible wavenumber \(k\in\mathcal K:=\frac{\pi}{L}\mathbb Z^N\), the Fourier coefficient of the periodized kernel is given by
\begin{equation}
    \widehat K(k)
    =
    \int_\Omega K(x)e^{-ik\cdot x}\,dx
    =
    \int_{\mathbb R^N}K_0(x)e^{-ik\cdot x}\,dx
    \label{eq:periodic_fourier_coeff}
\end{equation}
Thus, the discrete Fourier coefficients of the periodized kernel coincide with the full-space Fourier transform of the base kernel \(K_0\) evaluated at the admissible wavenumbers \(k\in\mathcal K\). This identity allows the effect of the periodized kernel to be treated algebraically in wavenumber space in both the linear stability analysis and the weakly nonlinear analysis.

In the examples considered in this paper, we mainly use a tent kernel. We nondimensionalize the system using the sensing radius $R$ as the reference length and set \(R=1\). The domain size $L$ then represents the size of the domain relative to the sensing radius. The base kernel associated with the tent kernel is radially symmetric and, with \(r=|x|\), is defined by
\begin{equation}
    K_{0,\mathrm{tent}}(x)
    =
    K_{0,\mathrm{tent}}(r)
    =
    c_N(1-r)_+
    \label{eq:tent_kernel}
\end{equation}
Here, \((s)_+=\max\{s,0\}\), and the constant \(c_N\) is chosen so that
\[
    \int_{\mathbb R^N}K_{0,\mathrm{tent}}(x)\,dx=1
\]
For example, when \(N=1\),
$
    c_1=1,
$
whereas when \(N=2\),
$
    c_2={3}/{\pi}
$
.
Using the adhesion-strength function
\[
    \omega_{\mathrm{tent}}(r)
    =
    -\frac{dK_{0,\mathrm{tent}}}{dr}(r)
    =
    c_N\mathbf 1_{(0,1)}(r)
\]
the nonlocal adhesion field can be written as
\[
    \nabla(K*u)(x)
    =
    \int_0^1\int_{\mathbb S^{N-1}}
        u(x+r\eta)\,
        \omega_{\mathrm{tent}}(r)\,
        r^{N-1}\eta
    \,d\eta\,dr
\]
This representation gives the tent kernel a simple mechanical interpretation: cells attract all other cells within the sensing radius \(R=1\) with the same strength, independently of their distance.

In the main text, we use this kernel as the primary example in both the analysis and the numerical simulations. Nevertheless, when the kernel shape is changed, the bifurcation structure can still be classified within the same framework through the second-harmonic response ratio introduced below. Other representative purely attractive kernels and their Fourier responses are summarized in \ref{app:kernels}.

\subsection{Known results on existence and boundedness}
\label{subsec:known_wellposedness}

Basic mathematical properties of weak solutions to the model \eqref{eq:main_model} have already been established by Ishii--Murakawa--Tanaka \cite{IshiiMurakawaTanaka2026}. In particular, for suitable initial data satisfying
\begin{equation}
    0\le u_0(x)\le 1
    \qquad
    \text{a.e. in } \Omega
    \label{eq:initial_bound}
\end{equation}
a corresponding weak solution exists, and the bound
\begin{equation}
    0\le u(x,t)\le 1
    \qquad
    \text{a.e. in } \Omega\times(0,T)
    \label{eq:solution_bound}
\end{equation}
is preserved throughout the evolution. This boundedness is essential to the physical interpretation of the model. The upper bound \(u=1\) represents the capacity limit at which cells completely occupy the available space, and states exceeding this limit have no biological meaning. By contrast, \(u=0\) represents a vacuum region in which no cells are present. Thus, preservation of the solution within the interval \([0,1]\) is a fundamental property ensuring the physical consistency of the model.

\subsection{Gradient-flow formulation and energy dissipation structure}
\label{subsec:energy_structure}

We next describe the energy structure of the model \eqref{eq:main_model}. We first define the entropy density \(F\) by
\begin{equation}
    F(u)
    =
    -\int_0^u \log(1-s^m)\,ds,
    \qquad
    0\le u\le 1
    \label{eq:F_def}
\end{equation}
where the value at $u=1$ is defined as an improper integral and is finite. Then
\begin{equation}
    F'(u)
    =
    -\log(1-u^m)
    \label{eq:F_prime}
\end{equation}
and \(F'(u)\) diverges as \(u\to1\). This singularity acts as an entropic repulsion as the density approaches the capacity limit.

We define the total energy functional by
\begin{equation}
    E[u]
    =
    \int_\Omega F(u)\,dx
    -
    \frac{a}{2}
    \int_\Omega u(K*u)\,dx
    \label{eq:energy}
\end{equation}
The first term is the local energy reflecting the population pressure and the capacity constraint, whereas the second term is the nonlocal attraction energy arising from cell-cell adhesion. Since the kernel \(K\) is obtained by periodizing a radially symmetric base kernel, it satisfies
$
K(x)=K(-x)
$. Under this symmetry, the formal first variation of the energy \eqref{eq:energy} is given by
\begin{equation}
    \mu
    :=
    \frac{\delta E}{\delta u}
    =
    F'(u)-aK*u
    =
    -\log(1-u^m)-aK*u
    \label{eq:chemical_potential}
\end{equation}

Using this chemical potential, the model \eqref{eq:main_model} can be rewritten in the following gradient-flow form:
\begin{equation}
    \frac{\partial u}{\partial t}
    =
    \nabla\cdot\left(g(u)\nabla\mu\right).
    \label{eq:gradient_flow_form}
\end{equation}

Here, $g(u):=u(1-u^m)$. Formally, using \eqref{eq:chemical_potential}, \eqref{eq:gradient_flow_form}, and the boundedness \eqref{eq:solution_bound}, and integrating by parts under periodic boundary conditions, we obtain the energy dissipation law
\[
\begin{aligned}
    \frac{d}{dt}E[u(t)]
    &=
    \int_\Omega
    \frac{\delta E}{\delta u}
    \frac{\partial u}{\partial t}\,dx \\
    &=
    \int_\Omega
    \mu\,\nabla\cdot(g(u)\nabla\mu)\,dx \\
    &=
    -\int_\Omega
    g(u)|\nabla\mu|^2\,dx \le 0.
\end{aligned}
\]
This formal energy structure plays an important role in the interpretation of the mesa limit for large \(m\) in \ref{app:mesa_limit_proofs}.
\section{Linear stability analysis and mode selection}
\label{sec:linear_stability}

In this section, we perform a linear stability analysis about a spatially homogeneous steady state and investigate the onset condition for pattern formation and the spatial modes selected during the initial stage. The nonlocal adhesion term tends to amplify density heterogeneities, whereas local degenerate diffusion tends to homogenize the density. The stability of the homogeneous state is therefore determined by the competition between these two effects.

First, for any constant \(\bar u\in[0,1]\), \(u(x,t)\equiv \bar u\) is a steady-state solution of \eqref{eq:main_model}. Indeed, if \(\bar u\) is spatially constant, then \(\nabla \bar u^m=0\) and \(\nabla(K*\bar u)=0\), so both the diffusive and adhesive fluxes vanish. In what follows, we assume 
$0<\bar u<1$ to consider a nontrivial linearization.

\subsection{Dispersion relation and instability threshold}
\label{subsec:dispersion_relation}

We introduce a small perturbation \(u(x,t)=\bar u+\varepsilon\widetilde u(x,t)\), \(0<\varepsilon\ll1\), about the homogeneous steady state \(u=\bar u\) and retain the terms of first order in \(\varepsilon\). For the diffusion term, we have
\[
    \nabla\cdot(u\nabla u^m)
    =
    \nabla\cdot\left(m u^m\nabla u\right)
\]
and hence the linearization coefficient is \(D(\bar u)=m\bar u^m\). For the adhesion term, the coefficient \(g(\bar u)=\bar u(1-\bar u^m)\) appears. Since \(\nabla(K*\bar u)=0\) for the constant state, the first-order term involving \(g'(\bar u)\) vanishes. The linearized equation is therefore
\begin{equation}
    \frac{\partial \widetilde u}{\partial t}
    =
    D(\bar u)\Delta\widetilde u
    -
    a g(\bar u)\Delta(K*\widetilde u)
    \label{eq:linearized_equation}
\end{equation}

We substitute the Fourier mode
\[
    \widetilde u(x,t)
    =
    \widetilde U e^{\lambda t+ik\cdot x},
    \qquad
    k\in\mathcal K
\]
into the linearized equation. By the properties of periodic convolution,
\[
    K*\widetilde u
    =
    \widehat K(k)\widetilde u,
    \qquad
    \Delta\widetilde u
    =
    -|k|^2\widetilde u
\]
and hence the growth rate \(\lambda(k)\) is given by
\begin{equation}
    \lambda(k)
    =
    |k|^2
    \left[
        a\bar u(1-\bar u^m)\widehat K(k)
        -
        m\bar u^m
    \right]
    \label{eq:dispersion_relation}
\end{equation}

The dispersion relation \eqref{eq:dispersion_relation} can be expressed more transparently in terms of the strength of nonlocal adhesion relative to local diffusion. We therefore define the effective adhesion ratio by
\begin{equation}
    \Gamma(\bar u)
    :=
    \frac{a\bar u(1-\bar u^m)}{m\bar u^m}
    =
    \frac{a(1-\bar u^m)}{m\bar u^{m-1}}
    \label{eq:effective_adhesion_ratio}
\end{equation}
Then \eqref{eq:dispersion_relation} can be written as
\begin{equation}
    \lambda(k)
    =
    m\bar u^m |k|^2
    \left[
        \Gamma(\bar u)\widehat K(k)-1
    \right]
    \label{eq:dispersion_gamma}
\end{equation}
Consequently, the condition for a nonzero wavenumber \(k\) to be linearly unstable is
\begin{equation}
\label{eq:cond_linearinstability}
        \Gamma(\bar u)\widehat K(k)>1
\end{equation}

Because the total mass is conserved, perturbations about a fixed mean density \(\bar u\) are constrained to have zero mean. The \(k=0\) component is therefore excluded from the Fourier expansion, and in what follows only the modes with \(k\ne0\) are regarded as pattern-forming modes. Thus, the homogeneous state is linearly unstable if and only if there exists a nonzero wavenumber \(k\in\mathcal K\setminus\{0\}\) such that
\[
    \lambda(k)>0
\]

\begin{theorem}[Linear instability threshold]
\label{thm:linear_threshold}
Consider the homogeneous steady state \(u=\bar u\), \(0<\bar u<1\). For the nonzero admissible wavenumbers, define
\[
    \widehat K_{\max}
    :=
    \max_{k\in\mathcal K\setminus\{0\}}\widehat K(k)
\]
and assume that \(\widehat K_{\max}>0\). Then the homogeneous steady state \(u=\bar u\) is linearly unstable if
\begin{equation}
    a>a_c(\bar u)
    :=
    \frac{m\bar u^{m-1}}
    {(1-\bar u^m)\widehat K_{\max}}
    \label{eq:critical_adhesion}
\end{equation}
and is linearly stable if
\[
    0<a<a_c(\bar u).
\]
\end{theorem}

\begin{figure*}[tbp]
    \centering
\includegraphics[width=0.75\linewidth]{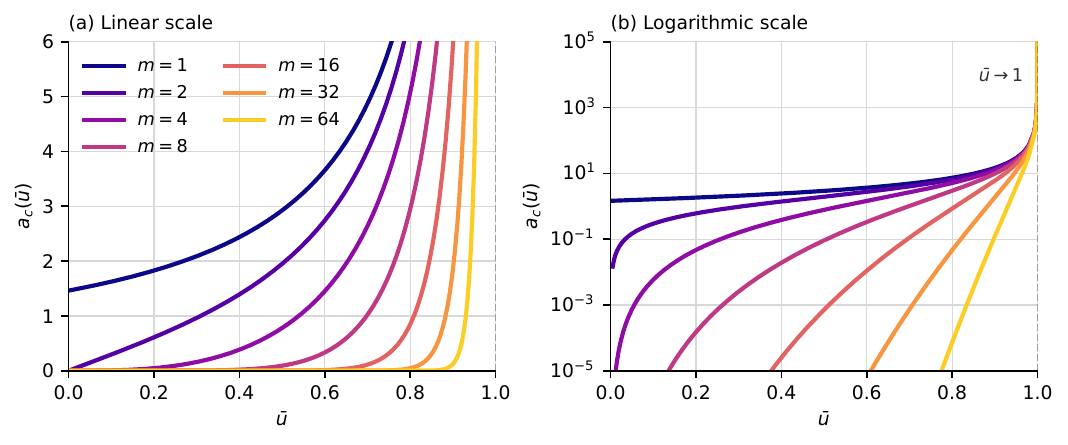}
\caption{
Dependence of the critical adhesion strength $a_c(\bar u)$ on the mean density $\bar u$ and the nonlinear exponent $m$.
The tent kernel in \eqref{eq:tent_kernel} is used, with
$L=1.5$ and $\Omega=[-L,L)$.
For the nonzero admissible wavenumbers,
$\widehat K_{\max}=\max_{k\in(\pi/L)\mathbb{Z}\setminus\{0\}}\widehat K(k)
=27/(4\pi^2)$, and
$a_c(\bar u)=m\bar u^{m-1}/\{(1-\bar u^m)\widehat K_{\max}\}$.
Panel (a) shows the dependence on $m$ at low to moderate mean densities on a linear scale, whereas panel (b) shows the divergence of $a_c(\bar u)$ as $\bar u\to1$ on a logarithmic scale.
}
    \label{fig:threshold_ac}
\end{figure*}

\begin{proof}
By the dispersion relation \eqref{eq:dispersion_relation}, the condition for a nonzero wavenumber \(k\) to become unstable is
\[
    a\bar u(1-\bar u^m)\widehat K(k)
    -
    m\bar u^m
    >0
\]
When \(\widehat K(k)>0\), this condition is equivalent to
\[
    a>
    \frac{m\bar u^{m-1}}
    {(1-\bar u^m)\widehat K(k)}
\]
Thus, the mode that first becomes unstable is the one that maximizes \(\widehat K(k)\) among the nonzero admissible wavenumbers, and the critical value is given by \eqref{eq:critical_adhesion}. Conversely, if \(a<a_c(\bar u)\), then \(\lambda(k)<0\) for every nonzero wavenumber, and the homogeneous state is therefore linearly stable.
\end{proof}

Equation \eqref{eq:critical_adhesion} demonstrates the stabilization caused by enhanced local dispersal and reduced adhesion mobility as the mean density increases. Indeed,
\[
    \bar u\to1
    \qquad\Longrightarrow\qquad
    a_c(\bar u)\to\infty
\]
This occurs because, as \(\bar u\to1\), the local diffusion coefficient \(m\bar u^m\) remains positive, whereas the adhesive-flux coefficient \(\bar u(1-\bar u^m)\) vanishes. Consequently, the homogeneous state is linearly stabilized.

\subsection{Effective adhesion ratio and density dependence of the fastest-growing mode}
\label{subsec:effective_attraction}

In this subsection, we investigate how the nonzero wavenumber with the largest growth rate changes with the mean density \(\bar u\) within the unstable regime. From the dispersion relation \eqref{eq:dispersion_gamma}, mode selection is determined by comparing
\[
    f_\Gamma(k)=|k|^2
    \left[
        \Gamma(\bar u)\widehat K(k)-1
    \right]
\]
over the nonzero admissible wavenumbers \(k\in\mathcal K\setminus\{0\}\).

We first examine the density dependence of the effective adhesion ratio \(\Gamma(\bar u)\). For fixed \(a>0\) and \(m\ge1\), \(\Gamma(\bar u)\) is strictly decreasing on \(0<\bar u<1\). Indeed, as the mean density increases, the volume-filling factor \(1-\bar u^m\) decreases, weakening the effectiveness of adhesion-driven advection. At the same time, the diffusion coefficient \(m\bar u^m\) increases, thereby strengthening the influence of diffusion relative to adhesion.

We next consider the growth rate \(f_\Gamma(q)\), treating the wavenumber as a continuous variable \(q\). Suppose that, over a continuous wavenumber interval
\(I=(\underline q,\overline q)\subset(0,\infty)\),
\(f_\Gamma\) attains its maximum in the interior of \(I\). At the maximizer \(q=q_*=q_*(\bar u)\), we have \(f_\Gamma'(q_*)=0\).
Direct calculation gives
\[
    f_\Gamma'(q)
    =
    q
    \left[
        \Gamma(\bar u)
        \left\{
            2\widehat K(q)+q\widehat K'(q)
        \right\}
        -2
    \right].
\]
Thus, defining
\begin{equation}
    H(q):=2\widehat K(q)+q\widehat K'(q)
    \label{def:H}
\end{equation}
an interior maximizer with \(q>0\) satisfies
\begin{equation}
    H(q)=\frac{2}{\Gamma(\bar u)}.
    \label{eq:H_peak_condition}
\end{equation}

The Fourier response of the tent kernel \eqref{eq:tent_kernel} is
$    \widehat K(q)
    =
    \frac{2(1-\cos(qR))}{(qR)^2}
$, and
$    H(q)
    =
    \frac{2\sin(qR)}{qR}
$.
Therefore, at least for \(0<qR<\pi\), \(H(q)\) is positive and strictly decreasing. As the mean density \(\bar u\) increases, \(\Gamma(\bar u)\) decreases and hence \(2/\Gamma(\bar u)\) increases. Consequently, the value of \(q\) satisfying the peak condition \eqref{eq:H_peak_condition} shifts toward lower wavenumbers. Thus, as the mean density increases, the fastest-growing mode shifts toward longer wavelengths.

The same mechanism applies to the other representative kernels considered in \ref{app:kernels}, at least for their dominant unstable modes at low wavenumbers.
This argument holds under the following sufficient conditions.

\begin{proposition}[Density-induced shift of the fastest-growing wavenumber]
\label{prop:density_longwave_shift}
Suppose that \(H(q)\) is positive and strictly decreasing on a continuous wavenumber interval \(I=(\underline q,\overline q)\subset(0,\infty)\).
Suppose further that, over the range of \(\bar u\) under consideration, \(f_\Gamma(q)\) has a unique maximum at \(q=q_*(\bar u)\in I\), characterized by \eqref{eq:H_peak_condition}. Then, for fixed adhesion strength \(a\), increasing the mean density \(\bar u\) decreases both the effective adhesion ratio \(\Gamma(\bar u)\) and the continuous wavenumber \(q_*(\bar u)\) associated with the maximum growth rate.
\end{proposition}

\if0
\begin{proof}
An interior maximizer satisfies \(H(q_*)=2/\\Gamma(\\bar u)\\). Since \\(\\Gamma(\\bar u)\\) is strictly decreasing on \\(0<\\bar u<1\\), increasing \\(\\bar u\\) increases \\(2/\\Gamma(\\bar u)\\). By assumption, \(H\) is strictly decreasing. Therefore, to maintain the equality \\(H(q\_*)=2/\Gamma(\bar u)\), \(q_*\) must decrease.
\end{proof}
\fi 

This proposition describes the shift in the peak position when the wavenumber is treated as a continuous variable.
On an actual periodic domain, however, the admissible wavenumbers are restricted to \(k_n=n\pi/L\), and therefore the continuous maximizer does not directly determine the selected mode. The shift of the fastest-growing mode toward lower wavenumbers as the mean density increases is observed as a decrease in the selected mode number \(n\).
The actual mode selection is determined by comparing \(\lambda(k_n)\) over the discrete set of admissible wavenumbers. Thus, when the continuous maximizer lies between two admissible wavenumbers, the neighboring discrete mode with the larger growth rate is selected.
Moreover, if two adjacent modes satisfy
$
    \lambda(k_n)=\lambda(k_{n+1})
$
and both attain the maximum growth rate, the two modes are simultaneously the fastest-growing modes. Such parameter values constitute switching points in mode selection and correspond to boundaries between regions in the mode-selection phase diagram presented below.

Fig.~\ref{fig:mode_selection_density} illustrates density-dependent mode selection for the tent kernel~\eqref{eq:tent_kernel} with \(L=5\) fixed. Panel (a) shows the linear growth rates obtained by varying the mean density \(\bar u\) while fixing the adhesion strength at \(a=4\). At low mean densities, relatively high-wavenumber modes have the largest growth rates, whereas the peak of the growth rate shifts toward lower wavenumbers as \(\bar u\) increases. Panel (b) presents the mode-selection phase diagram in the \((a,\bar u)\) plane, showing that, for fixed \(a\), the selected mode decreases stepwise as the mean density \(\bar u\) increases. The diagram also shows the stepwise changes in the selected mode when the adhesion strength \(a\) is varied at fixed \(\bar u\). Panels (c)--(f) show time-dependent simulations of the nonlinear dynamics at \(a=4\) for the representative mean densities indicated in panels (a) and (b).
Each panel displays three profiles: the initial profile, the profile at the time when the linearly unstable mode has grown and become clearly visible, and the profile after a sufficiently long time, when a numerical steady state has been reached. The unstable mode predicted by the linear theory is selected and grows during the initial stage, after which it develops into a nonlinear aggregation pattern. The numerical method used in this paper is described in \ref{app:numerical_scheme}.

It should be noted that the mode predicted by the linear analysis determines only the characteristic spatial scale during the initial stage of the evolution and does not necessarily determine the number of aggregates in the final steady state. In particular, when aggregation proceeds slowly or when the tails of the resulting aggregates lie within the sensing range of their mutual interaction, that is, within the support of the nonlocal kernel, neighboring aggregates may attract one another, leading to merging or coarsening.

\begin{figure*}[t]
    \centering
    \includegraphics[width=\linewidth]{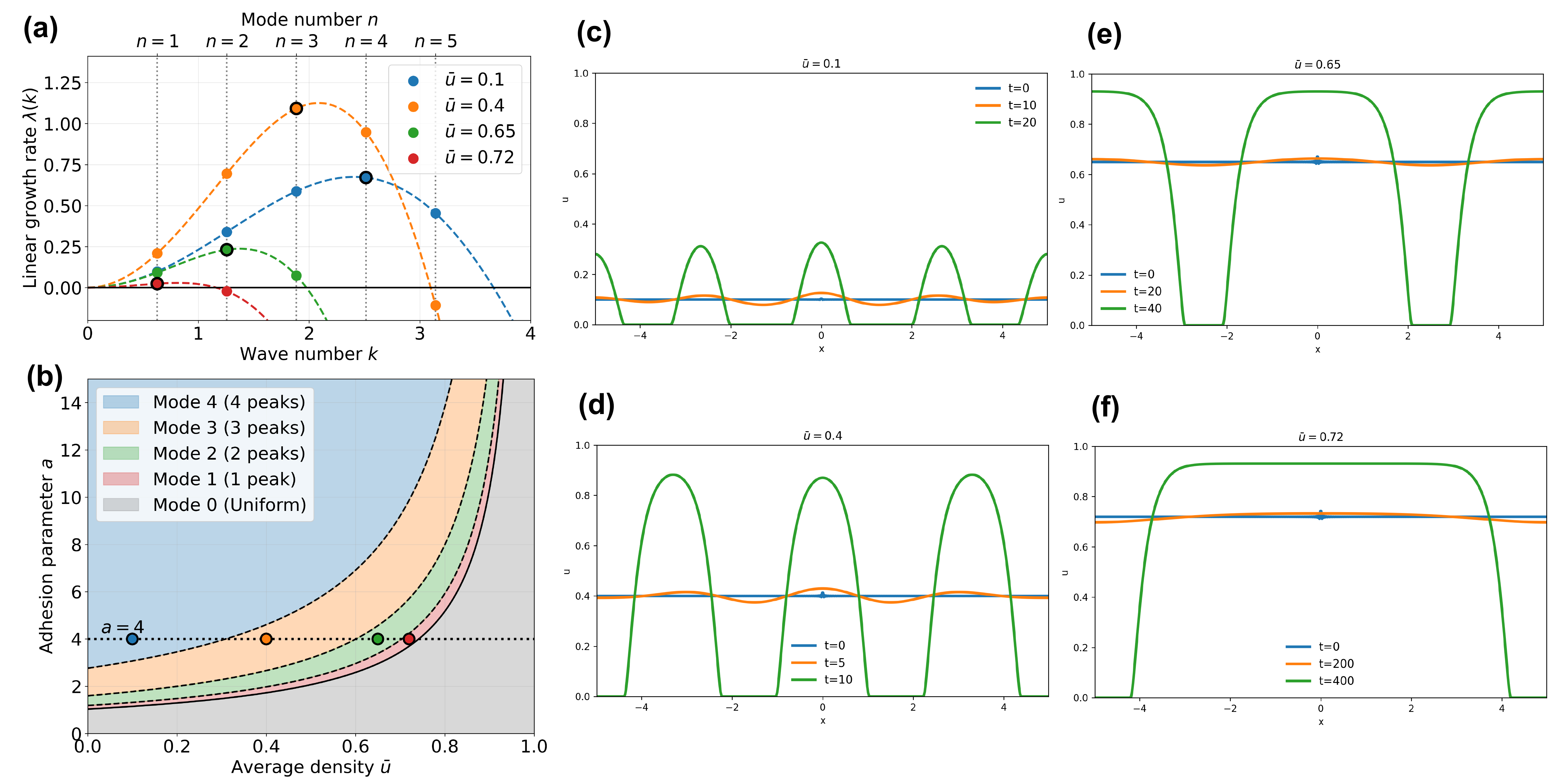}
    \caption{
    Shift of the fastest-growing mode toward lower wavenumbers with increasing mean density for the tent kernel~\eqref{eq:tent_kernel}, with \(L=5\) fixed.
    (a) Linear growth rate \(\lambda(k)\) for different mean densities \(\bar u\), with the adhesion strength fixed at \(a=4\).
    The dotted curves show the growth rate when \(k\) is treated as a continuous variable, and the circles indicate the admissible modes \(k_n=n\pi/L\) on the periodic domain. The black circles indicate the fastest-growing modes.
    As \(\bar u\) increases, the mode with the maximum growth rate shifts toward lower wavenumbers.
    (b) Linear mode-selection phase diagram in the \((a,\bar u)\) plane.
    The horizontal line indicates \(a=4\), and the black circles indicate the mean densities used in panels (a) and (c)--(f).
    (c)--(f) Time evolution of the numerical solutions for \(\bar u=0.1,0.4,0.65,0.72\), respectively.
    }
    \label{fig:mode_selection_density}
\end{figure*}

\section{Weakly nonlinear analysis and bifurcation classification}
\label{sec:weakly_nonlinear}

In the preceding section, we derived the condition for linear instability of the homogeneous steady state and the selection rule for the modes that grow during the initial stage. Linear analysis, however, provides only the initial growth rates of perturbations and cannot determine whether the resulting pattern saturates smoothly at small amplitude or undergoes an abrupt transition to large amplitude. To determine the bifurcation type, we perform a weakly nonlinear analysis near the critical adhesion strength and derive an amplitude equation for the critical mode.

In what follows, we consider a one-dimensional periodic domain and assume that, as the adhesion strength \(a\) is increased, the only critical wavenumbers at which the homogeneous state first becomes unstable are \(k=\pm k_c\).
More precisely, we assume that, at the critical value \(a=a_c\),
\[
    \lambda(\pm k_c)=0,
    \qquad
    \lambda(k)<0
    \quad
    (k\in\mathcal K\setminus\{0,\pm k_c\})
\]
holds. Under this assumption, the bifurcation near the critical point is described by a single complex amplitude. On multidimensional periodic domains, several wavenumbers of the same magnitude may become critical simultaneously, in which case multimode amplitude equations are required. In this paper, we focus primarily on one spatial dimension and analyze the case of a simple critical mode.

\subsection{Derivation of the Stuart--Landau equation}
\label{subsec:SL_reduction}

We expand the adhesion strength near its critical value as
\begin{equation}
    a=a_c+\varepsilon^2a_2,
    \qquad
    0<\varepsilon\ll1
    \label{eq:a_expansion}
\end{equation}
and introduce the slow time variable \(\tau=\varepsilon^2t\). We also expand the solution about the homogeneous steady state \(\bar u\) as
\begin{equation}
    u(x,t)
    =
    \bar u
    +
    \varepsilon u_1(x,\tau)
    +
    \varepsilon^2u_2(x,\tau)
    +
    \varepsilon^3u_3(x,\tau)
    +
    O(\varepsilon^4)
    \label{eq:weak_expansion}
\end{equation}
Under the assumption that the critical mode is simple, the first-order term can be written as
\begin{equation}
    u_1(x,\tau)
    =
    A(\tau)e^{ik_cx}
    +
    \overline{A(\tau)}e^{-ik_cx}
    \label{eq:u1_def}
\end{equation}
where \(A(\tau)\) is a complex amplitude depending on the slow time.

A standard multiple-scale expansion, or a center-manifold reduction, shows that the amplitude \(A(\tau)\) of the critical mode satisfies the Stuart--Landau equation
\begin{equation}
    \frac{dA}{d\tau}
    =
    \sigma A
    -
    c_{\mathrm L}|A|^2A
    \label{eq:SL_equation}
\end{equation}
Here, \(\sigma\) is the linear growth coefficient proportional to the deviation from the critical value and, for the present model,
$    \sigma=a_2\bar u(1-\bar u^m)\widehat K(k_c)k_c^2
$.
Moreover, \(c_{\mathrm L}\) is the Landau coefficient, which determines the bifurcation type.

In the Stuart--Landau equation, \(A=0\) is the trivial solution corresponding to the homogeneous steady state.
By contrast, a nontrivial steady solution with \(A\neq0\) satisfies
\[
    |A|^2=\frac{\sigma}{c_{\mathrm L}}
\]
Therefore, when \(c_{\mathrm L}>0\), a branch of nontrivial small-amplitude steady solutions emerges on the \(\sigma>0\) side and connects continuously to the trivial solution \(A=0\) as \(\sigma\to0+\). Since a branch of small-amplitude patterned solutions bifurcates from the homogeneous state at \(\sigma=0\), we refer to this case as a supercritical bifurcation.
When \(c_{\mathrm L}<0\), by contrast, the branch of nontrivial small-amplitude steady solutions emerges on the \(\sigma<0\) side. Consequently, small-amplitude saturation does not occur on the linearly unstable side \(\sigma>0\), and solutions beyond the critical point are expected to transition to the large-amplitude regime. We refer to this case as a subcritical bifurcation.

\subsection{Explicit expression for the Landau coefficient and bifurcation criterion}
\label{subsec:explicit_landau}

To obtain an explicit expression for the Landau coefficient, we introduce the diffusion potential
\[
    P(u):=\frac{m}{m+1}u^{m+1}
\]
Then \(\nabla\cdot(u\nabla u^m)=\Delta P(u)\). We introduce the following coefficients:
\begin{align}
    D_1 &= P'(\bar u)=m\bar u^m, \notag\\
    D_2 &= \frac12P''(\bar u)=\frac12m^2\bar u^{m-1}, \notag\\
    D_3 &= \frac16P'''(\bar u)
          =\frac16m^2(m-1)\bar u^{m-2},
          \label{eq:D_coefficients}\\
    G_0 &= g(\bar u)=\bar u(1-\bar u^m), \notag\\
    G_1 &= g'(\bar u)=1-(m+1)\bar u^m, \notag\\
    G_2 &= \frac12g''(\bar u)
          =-\frac12m(m+1)\bar u^{m-1}.
          \label{eq:G_coefficients}
\end{align}
At second order in \(\varepsilon\), the quadratic nonlinear interaction of the critical mode \(e^{ik_cx}\) generates the second harmonic
\(e^{2ik_cx}\).
Because the present model is in divergence form, no mean component arises at this order.
Therefore, the second-order term can be written using a real coefficient \(U_2\) as
\[
    u_2
    =
    U_2A^2e^{2ik_cx}
    +
    {U_2}\,\overline A^2e^{-2ik_cx}
\]
Furthermore, since the second harmonic \(2k_c\) is noncritical, substituting this \(u_2\) into the second-order equation in \(\varepsilon\) determines \(U_2\):
\begin{equation}
    U_2
    =
    \frac{1}{2D_1(\kappa-1)}
    \left(
        2D_2-\frac{D_1G_1}{G_0}
    \right).
    \label{eq:U2_formula}
\end{equation}
Here, the second-harmonic response ratio \(\kappa\) is defined by
\begin{equation}
    \kappa
    :=
    \frac{\widehat K(2k_c)}{\widehat K(k_c)}
    \label{eq:kappa_def}
\end{equation}
Requiring the critical-mode component to vanish in the third-order equation in \(\varepsilon\) determines the Landau coefficient \(c_{\mathrm L}\) appearing in equation \eqref{eq:SL_equation} as follows:
\begin{align}
        c_{\mathrm L}
    &=
    k_c^2
    \left[
        U_2
        \left(
            2D_2
            -
            \frac{D_1G_1}{G_0}(2\kappa-1)
        \right)
        +
        3D_3
        -
        \frac{D_1G_2}{G_0}
    \right]\\
        &=
    \frac{m\bar u^{m-2}k_c^2}{2}
    \mathcal B(m,\bar u,\kappa).
    \label{eq:landau_constant}
\end{align}
Here,
\begin{equation}
    \begin{aligned}
        \mathcal B(m,\bar u,\kappa)
        &:=
        \underbrace{
            \frac{(m-\Phi)\bigl(m-\Phi(2\kappa-1)\bigr)}
            {\kappa-1}
        }_{\text{quadratic nonlinear contribution (from \(U_2\))}}
        \\[-1mm]
        &\quad+
        \underbrace{
            m\frac{m-1+2\bar u^m}{1-\bar u^m}.
        }_{\text{cubic nonlinear contribution (positive)}}
    \end{aligned}
    \label{eq:B_def}
\end{equation}
\[
    \Phi:=
    \frac{1-(m+1)\bar u^m}{1-\bar u^m}.
\]
The appearance of \(\kappa\) in the Landau coefficient \(c_L\) shows that the extent to which the second harmonic \(2k_c\) is attenuated by the kernel relative to the fundamental mode \(k_c\) plays an important role in the weakly nonlinear analysis.
Details of the coefficient calculations are provided in \ref{app:weakly_nonlinear_details}.

In \eqref{eq:landau_constant}, the coefficient multiplying \({\mathcal B}\) is positive for
\(0<\bar u<1\) and \(m\ge1\).
Therefore, the sign of the Landau coefficient \(c_{\mathrm L}\) is determined by the sign of \({\mathcal B}\).

Equation \eqref{eq:B_def} represents the competition between the first term, which arises from the quadratic nonlinear interaction, and the second term, which arises from the cubic self-interaction.
Since the second term is positive, it acts as a stabilizing effect that drives the Landau coefficient \(c_{\mathrm L}\) in the positive direction.
For the first term, \(m-\Phi>0\) under the degenerate-diffusion condition \(m\ge1\). When the second-harmonic response ratio satisfies \(\kappa<1\), the denominator \(\kappa-1\) is negative, whereas the sign of the remaining factor
$
    m-\Phi(2\kappa-1)
$
depends on the parameters.

In particular, if \(m-\Phi(2\kappa-1)>0\), the first term drives \({\mathcal B}\) in the negative direction, and the quadratic nonlinear effect acts to promote a subcritical bifurcation.
By contrast, if \(m-\Phi(2\kappa-1)<0\), the first term also increases \({\mathcal B}\), and therefore acts in the same direction as the cubic stabilizing effect.

\begin{proposition}[Explicit Landau coefficient]
\label{thm:explicit_landau}
Under the assumption of a simple critical mode, suppose that the second-harmonic response ratio satisfies \(0\le \kappa<1\).
Then the Landau coefficient appearing in the Stuart--Landau equation \eqref{eq:SL_equation} is given by \eqref{eq:landau_constant}.
Moreover, the critical bifurcation is supercritical if \(\mathcal B(m,\bar u,\kappa)>0\) and subcritical if \(\mathcal B(m,\bar u,\kappa)<0\).
\end{proposition}

\begin{remark}
The adhesion strength \(a\) does not appear explicitly in the Landau coefficient \(c_{\mathrm L}\).
This property, however, is specific to the local bifurcation criterion near the critical point. In parameter regimes away from the critical point, the amplitude and shape of the resulting patterns may depend on \(a\).
\end{remark}

\subsection{Transition to subcriticality for large nonlinear exponents}
\label{sec:large_m_subcriticality}

\begin{figure*}[t]
    \centering
    \includegraphics[width=0.75\linewidth]{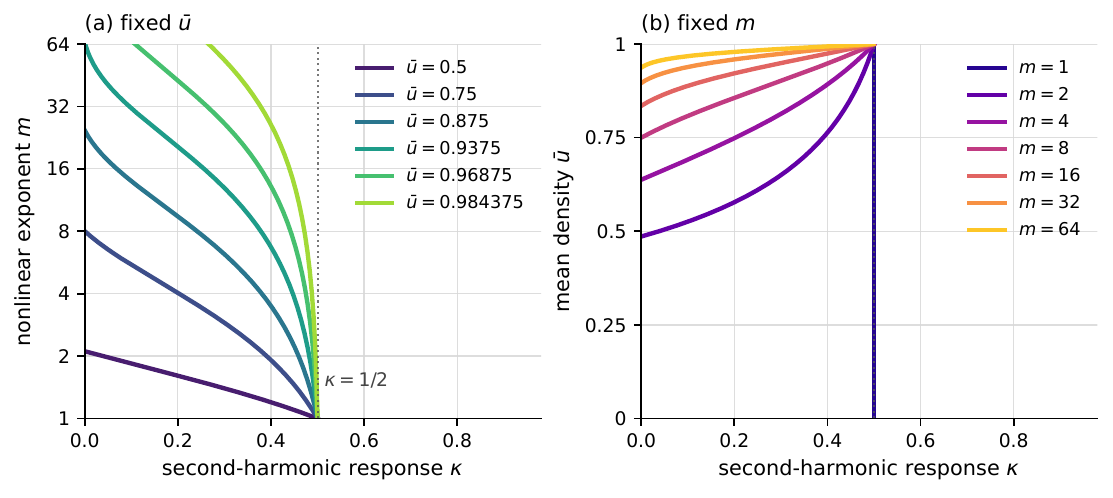}
    \caption{
    Critical curves \(\mathcal B(m,\bar u,\kappa)=0\) across which the Landau coefficient \(c_{\mathrm L}\) changes sign.
    (a) Critical curves in the \((\kappa,m)\) plane for \(\bar u=1/2,3/4,7/8,15/16,31/32,63/64\).
    (b) Critical curves in the \((\kappa,\bar u)\) plane for \(m=1,2,4,8,16,32,64\).
    The dashed line indicates \(\kappa=1/2\).
    }
    \label{fig:landau_classifier}
\end{figure*}

As shown in Fig.~\ref{fig:landau_classifier}, as the nonlinear exponent \(m\) increases, the range of the second-harmonic response ratio \(\kappa\) over which a supercritical bifurcation occurs becomes narrower.
In this subsection, we use \(\mathcal B\) to justify this tendency mathematically.

\begin{theorem}[Large-\(m\)-induced subcriticality]
\label{thm:large_m_obstruction}
For every mean density \(\bar u\in(0,1)\), there exists
\(m_c=m_c(\bar u)\ge1\) such that, if \(m>m_c\), then
\[
    {\mathcal B}(m,\bar u,\kappa)<0
\]
for every \(0\le\kappa<1\).
Consequently, the critical bifurcation is subcritical throughout this range.
\end{theorem}

\begin{proof}
Since \(\mathcal B(m,\bar u,\kappa)\) is monotonically decreasing with respect to \(\kappa\), for every \(0\le\kappa<1\),
\[
    \mathcal B(m,\bar u,\kappa)
    \le
    \mathcal B(m,\bar u,0).
\]
For fixed \(\bar u\in(0,1)\), as \(m\to\infty\),
\(v=\bar u^m\to0\) and \(\Phi\to1\). Therefore,
\[
    \mathcal B(m,\bar u,0)
    =
    -m+O(1)
    \qquad
    (m\to\infty).
\]
Hence, there exists \(m_c(\bar u)\) such that, whenever \(m>m_c(\bar u)\),
\[
    \mathcal B(m,\bar u,0)<0.
\]
Combining this inequality with the monotonicity estimate proves the result.
\end{proof}

\subsection{System-size dependence for the tent kernel}
\label{subsec:tent_system_size}

We examine how the second-harmonic response ratio \(\kappa\) varies with the system size for the tent kernel.
For the one-dimensional tent base kernel, the Fourier response is given by
$
    \widehat K(k)
    =
    \frac{2(1-\cos(kR))}{(kR)^2}
$.
Moreover, when the fundamental mode is selected as the critical mode on a periodic domain, \(k_c=\pi/L\). The second-harmonic response ratio is then
\begin{equation}
    \kappa_{\rm tent}
    =
    \frac{\widehat K(2k_c)}{\widehat K(k_c)}
    =
    \cos^2\left(\frac{k_cR}{2}\right)
=
    \cos^2\left(\frac{\pi R}{2L}\right)
    \label{eq:kappa_tent_system_size}
\end{equation}

In this paper, we nondimensionalize the system using the sensing radius as the reference length and set \(R=1\). In this case, \(\kappa_{\rm tent}=1/2\) when \(L=2\). The weakly nonlinear analysis in Subsection~\ref{subsec:explicit_landau} gives
\(\mathcal{B}(1,\frac{1}{2},\frac{1}{2})=0\) for \(m=1\) and \(\bar u=1/2\). Thus, \(L=2\) marks the boundary between supercritical and subcritical bifurcations.

The results of a numerical bifurcation analysis confirming the above calculation are shown in Fig.~\ref{fig:tent_system_bifurcation}. Here, \(m=1\), \(\bar u=1/2\), and \(R=1\). For \(L=1.9\),
\[
    \kappa_{\rm tent}
    =
    0.4587<\frac12,
    \qquad
    c_{\mathrm L}=0.4171>0
\]
and the weakly nonlinear analysis predicts a supercritical bifurcation. Indeed, in Fig.~\ref{fig:tent_system_bifurcation} (a), the nonhomogeneous branch bifurcating from the homogeneous branch extends into the \(a>a_c\) regime. By contrast, for \(L=2.1\),
\[
    \kappa_{\rm tent}
    =
    0.5374>\frac12,
    \qquad
    c_{\mathrm L}=-0.3615<0
\]
and the weakly nonlinear analysis predicts a subcritical bifurcation. In Fig.~\ref{fig:tent_system_bifurcation} (b),  the small-amplitude nonhomogeneous branch extends into the \(a<a_c\) regime and, after passing through a fold point, connects to the large-amplitude branch.

Thus, for the tent kernel, the system size \(L\) controls the bifurcation type through the second-harmonic response ratio \(\kappa\).

\begin{figure*}[t]
    \centering
    \includegraphics[width=\linewidth]{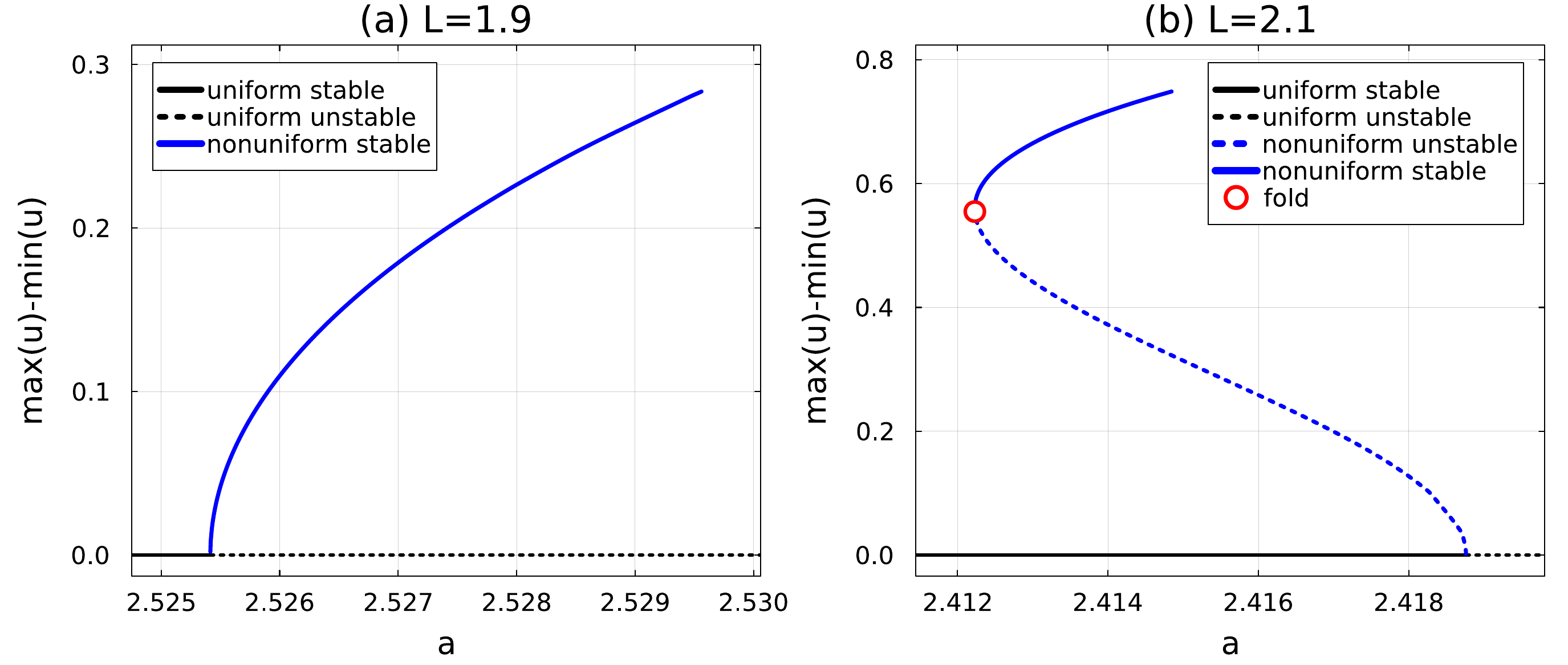}
    \caption{
    Bifurcation diagrams for the tent kernel with \(m=1\), \(\bar u=1/2\), and \(R=1\): (a) \(L=1.9\) and (b) \(L=2.1\).
    The horizontal axis represents the adhesion strength \(a\), and the vertical axis represents the magnitude of the spatial heterogeneity, \(\max(u)-\min(u)\).
    The black curves represent the homogeneous branch, and the blue curves represent the nonhomogeneous branch.
    Solid curves indicate stable branches, whereas dashed curves indicate unstable branches. The red open circle indicates the fold point.
    }
    \label{fig:tent_system_bifurcation}
\end{figure*}

The bifurcation types for other purely attractive kernels can also be classified within the same framework because differences between the kernels are reflected in \(\kappa\).
The Fourier responses of the kernels and the corresponding evaluations of \(\kappa\) are summarized in \ref{app:kernels}, while the numerical bifurcation method is described in \ref{app:numerical_bifurcation}.

\section{Mesa limit and capacity-constrained energy structure}
\label{sec:mesa_limit}

In this section, we interpret the mesa profiles that emerge for large values of the nonlinear exponent \(m\) in terms of the minimization problem for the limiting energy obtained from the \(m\to\infty\) limit (mesa limit) of the total energy \eqref{eq:energy}.
The \(m\to\infty\) limit of the free energy for nonlinear nonlocal aggregation--diffusion equations was analyzed as a mesa limit by Carrillo--Gvalani~\cite{CarrilloGvalani2021}.
For porous-medium-type free energies, they established \(\Gamma\)-convergence and characterized the minimizers of the corresponding limiting energy.
In this section, we apply a similar approach to the total energy \eqref{eq:energy} of the present model.

In what follows, we write the entropy density defined in \eqref{eq:F_def} as \(F_m\) to make its dependence on \(m\) explicit.
Similarly, we denote the total energy functional defined in \eqref{eq:energy} by \(E_m\).

In view of its interpretation as a cell density, we consider the energy structure of the present model under the capacity constraint \(0\le u\le1\).
We fix the total mass
$
    M=\int_\Omega u_0(x)\,dx
    =
    \bar u|\Omega|
$.
We define the space of functions with total mass \(M\) by
\[
    \mathcal P_M(\Omega)
    :=
    \left\{
        u\in L^1(\Omega)
        \,\middle|\,
        \int_\Omega u\,dx=M
    \right\},
\]
and the admissible set satisfying the capacity constraint by
\begin{equation}
    \mathcal A_M
    :=
    \left\{
        u\in L^\infty(\Omega)\cap \mathcal{P}_M(\Omega):
        0\le u\le1 \ \text{a.e. in } \Omega
    \right\}.
    \label{eq:admissible_set_mesa}
\end{equation}
Here, \(0<M\le|\Omega|\).
We extend \(F_m\) by setting \(F_m(u):=+\infty\) for \(u\notin[0,1]\).
Consequently, \(E_m[u]=+\infty\) whenever \(u\notin\mathcal A_M\).

It follows from \eqref{eq:local_entropy_uniform} that the local entropy term vanishes on the admissible set \(\mathcal A_M\):
\begin{equation}
    \int_\Omega F_m(u)\,dx
    \le \frac{\pi^2}{6m}|\Omega|
    \to 0 \quad \text{as $m \to \infty$}.
    \label{eq:vanishing_entropy_term}
\end{equation}
This suggests that the limiting energy consists solely of the nonlocal attraction term.
Indeed, the following theorem holds.

\begin{theorem}[Mesa limit of the energy]
\label{thm:mesa_energy_limit}
Suppose that \(K\) is even and \(K\in L^1(\Omega)\).
Define the limiting energy by
\begin{equation}
    E_\infty[u]
    :=
    \begin{cases}
    -\dfrac{a}{2}\displaystyle\int_\Omega u(K*u)\,dx,
    & u\in\mathcal A_M,\\[3mm]
    +\infty,
    & u\notin\mathcal A_M.
    \end{cases}
    \label{eq:E_infty_def}
\end{equation}
Then \(E_m\) converges uniformly to \(E_\infty\) on \(\mathcal A_M\):
\begin{equation}
    0
    \le
    E_m[u]-E_\infty[u]
    \le
    \frac{\pi^2|\Omega|}{6m}
    \to 0
    \quad\text{as }m\to\infty.
    \label{eq:uniform_energy_bound}
\end{equation}
In particular, the energies \(E_m\), defined on
\(L^\infty(\Omega)\cap\mathcal P_M(\Omega)\), \(\Gamma\)-converge to \(E_\infty\) with respect to the weak-\(*\) topology of \(L^\infty(\Omega)\) as \(m\to\infty\).
\end{theorem}

\noindent
The proof is provided in \ref{app_subsec:gamma_convergence_proof}.
As in the mesa limit of porous-medium-type free energies studied by Carrillo--Gvalani~\cite{CarrilloGvalani2021}, Theorem~\ref{thm:mesa_energy_limit} shows that the local entropy term vanishes in the limit, leaving a capacity-constrained nonlocal attraction energy.
The origin of the capacity constraint, however, is different.
In the setting of Carrillo--Gvalani~\cite{CarrilloGvalani2021}, the capacity constraint is not imposed for finite \(m\); instead, it emerges in the limit \(m\to\infty\) because the local entropy term excludes high-density regions.
In the present model, by contrast, the capacity constraint \(0\le u\le1\) is incorporated from the outset as a physical constraint on the cell density, even for finite \(m\).
Thus, the local entropy term in the present model does not generate the capacity constraint; rather, it vanishes in the limit on an already constrained admissible set.
Combining this observation with \eqref{eq:vanishing_entropy_term} yields both the uniform convergence of the total energies on \(\mathcal A_M\) and their \(\Gamma\)-convergence.

We conclude with a comment on the assumptions imposed on the kernel in Theorem~\ref{thm:mesa_energy_limit}.
In this study, kernels defined on the torus are constructed by periodizing symmetric kernels defined on the whole space.
The resulting periodized kernels therefore inherit the symmetry and boundedness of the original kernels.
For example, all the kernels listed in Table~1 of \ref{app:kernels} are symmetric and bounded.
Consequently, the corresponding kernels on the torus are even and belong to \(L^1(\Omega)\).

\subsection{Limiting minimizers and characteristic-function-type structure}
\label{subsec:mesa_minimizers}

We first establish the existence of a minimizer of the total energy \(E_m\) over \(\mathcal A_M\) for each finite \(m\).

\begin{proposition}[Existence of minimizers for fixed \(m\)]
\label{prop:minimizer_for_m}
Suppose that \(K\in L^1(\Omega)\).
Then, for every \(m\ge1\), the total energy functional
\(E_m[u]\)
admits a minimizer over \(\mathcal A_M\).
\end{proposition}

\noindent
The proof is provided in \ref{app_subsec:minimizer_convergence_proof}.

It follows from Theorem~\ref{thm:mesa_energy_limit} that, after passing to a subsequence, minimizers of the total energy \(E_m\) converge to a minimizer of the limiting energy \(E_\infty\).
This is a standard consequence of the theory of \(\Gamma\)-convergence; see Rindler~\cite{rindler}.

\begin{corollary}[Convergence of minimizers]
\label{cor:mesa_minimizer_convergence}
Let \(u_m\in\mathcal A_M\) be a minimizer of \(E_m\).
Then there exist a subsequence \(m_j\) and a function \(u_\infty\in\mathcal A_M\) such that
\[
    u_{m_j}
    \rightharpoonup^\ast
    u_\infty
    \qquad
    \text{in } L^\infty(\Omega) \quad \text{as $j \to \infty$}.
\]
Moreover, \(u_\infty\) is a minimizer of \(E_\infty\).
\end{corollary}

\noindent
The proof is provided in \ref{app_subsec:minimizer_convergence_proof}.

We next investigate the shape of the minimizers of the limiting energy \(E_\infty\).
Minimizing \(E_\infty\) is equivalent to maximizing the nonlocal attraction energy
$
    \int_\Omega u(K*u)\,dx
$
over the capacity-constrained admissible set \(\mathcal A_M\).
In particular, we consider the case in which the Fourier coefficients are nonnegative, as for the one-dimensional tent kernel:
\[
    \widehat K(k)\ge0
    \qquad
    (k\in\mathcal K).
\]
Under this condition, the nonlocal attraction energy is a convex quadratic form.
The convexity of this energy implies that the minimization problem for the limiting energy \(E_\infty\) admits a characteristic-function-type minimizer.

\begin{proposition}[Existence of a characteristic-function-type minimizer]
\label{prop:characteristic_mesa}
Suppose that \(K\in L^1(\Omega)\) and that the nonlocal attraction energy satisfies
\[
    \int_\Omega u(K*u)\,dx\ge0
    \qquad
    \text{for all } u\in L^2(\Omega).
\]
Then the minimization problem for the limiting energy \(E_\infty\) admits a minimizer of the form
\[
    u_\infty=\mathbf 1_A,
    \qquad
    |A|=M
\]
for some measurable set \(A\subset\Omega\).
\end{proposition}

\noindent
The proof is provided in \ref{app_subsec:characteristic_mesa_proof}.
Proposition~\ref{prop:characteristic_mesa} shows that the limiting problem admits at least one characteristic-function-type minimizer.
This variational structure is consistent with the mesa profiles observed numerically, which consist of saturated and background regions.

Furthermore, minimizers of the limiting energy satisfy the following threshold condition.

\begin{proposition}[Threshold condition]
\label{prop:mesa_threshold_condition}
Suppose that \(K\) is even and belongs to \(L^1(\Omega)\).
Let \(u_\infty\in\mathcal A_M\) be a minimizer of \(E_\infty\).
Then there exists a constant \(\lambda\in\mathbb R\) such that
\[
    K*u_\infty \ge \lambda
    \quad
    \text{a.e. on } \{0<u_\infty\le1\},
\]
and
\[
    K*u_\infty \le \lambda
    \quad
    \text{a.e. on } \{0\le u_\infty<1\}.
\]
Moreover, if the set on which \(0<u_\infty<1\) has positive measure, then
\[
    K*u_\infty=\lambda
\]
on that set.
\end{proposition}

\noindent
The proof is provided in \ref{app_subsec:threshold_condition_proof}.
This threshold condition corresponds to the Euler--Lagrange condition for the limiting energy under the capacity and fixed-mass constraints, or, more precisely, to the associated variational inequality.
Regions in which the nonlocal potential exceeds the threshold are filled up to the capacity limit and form saturated regions, whereas regions in which the potential lies below the threshold form background regions with zero density.
If a region of intermediate density exists, the nonlocal potential equals the threshold throughout that region.
\subsection{Mesa pattern formation with increasing nonlinear exponent \(m\)}
\label{subsec:numerical_mesa}

In this subsection, we investigate how solutions of model \eqref{eq:main_model} transition from smooth aggregation patterns to mesa profiles as the nonlinear exponent \(m\) increases.
We use the tent kernel \eqref{eq:tent_kernel} and fix the mean density at \(\bar u=0.6\), the adhesion strength at \(a=4\), the domain size at \(L=1.5\), and the sensing radius at \(R=1\).
Numerical simulations are performed by varying only the nonlinear exponent over \(m=2^2,2^4,2^6,2^8,2^{14}\).
The numerical solution profiles after sufficiently long evolution are shown in Fig.~\ref{fig:mesa_varying_m}.

\begin{figure*}[tbp]
    \centering
    \includegraphics[width=0.78\linewidth]{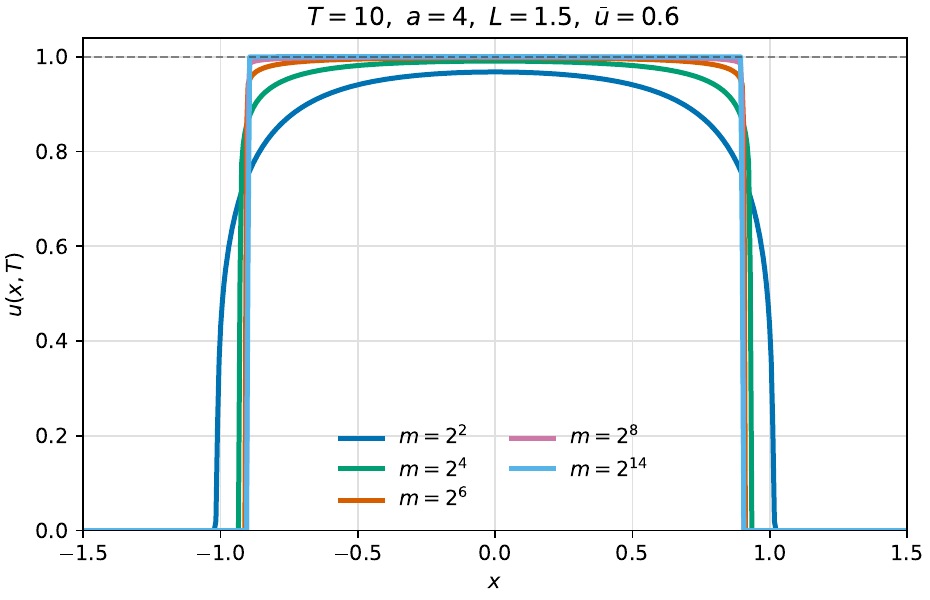}
    \caption{
    Mesa pattern formation with increasing nonlinear exponent \(m\).
    The parameters are \(\bar u=0.6\), \(a=4\), \(L=1.5\), and \(R=1\), and the profiles \(u(x,T)\) are shown at \(T=10\).
    The dashed line indicates the capacity limit \(u=1\).
    }
    \label{fig:mesa_varying_m}
\end{figure*}

For \(m=2^2\), the resulting aggregate has a relatively smooth, mound-shaped profile, and its maximum density remains below the capacity limit \(u=1\).
This occurs because the volume-filling effect is relatively weak, and aggregation driven by nonlocal attraction is balanced by the smoothing effect of nonlinear diffusion.

As \(m\) increases, \(m u^m\) tends to zero for every fixed \(0\le u<1\).
Consequently, in density regions away from the capacity limit, local dispersal generated by the population pressure becomes weaker, and aggregation driven by nonlocal adhesion becomes relatively dominant.
The numerical simulations therefore show that the density rises toward the capacity limit and that the solution approaches a mesa profile with an almost flat plateau and steep transition layers along its edges.

Mass conservation also strongly constrains the profile in the mesa limit.
For the parameter setting used in Fig.~\ref{fig:mesa_varying_m}, the domain length is \(2L=3\), and the mean density is \(\bar u=0.6\), giving a total mass of \(M=1.8\).
Thus, for a limiting profile consisting of a saturated region with \(u\simeq1\) and a background region with \(u\simeq0\), the total width of the saturated region is approximately \(1.8\).
The flat plateau and sharp interfaces observed for large \(m\) therefore reflect both the capacity constraint and mass conservation.

It should be emphasized that Theorem~\ref{thm:mesa_energy_limit} and Proposition~\ref{prop:characteristic_mesa} concern energy minimization problems.
They do not assert that a time-dependent solution at finite time necessarily reaches a global minimizer.

\FloatBarrier

\section{Conclusions}
\label{sec:conclusion}

In this paper, we analyzed pattern formation from a homogeneous steady state in a one-component nonlocal adhesion model with population pressure and degenerate mobility.
Through linear stability analysis, we derived the instability threshold and the density dependence of the fastest-growing mode, and identified the mechanism by which the selected wavenumber shifts toward lower wavenumbers as the mean density increases.

We also performed a weakly nonlinear analysis near the critical adhesion strength and showed that the amplitude of the critical mode satisfies a Stuart--Landau equation.
The Landau coefficient was expressed explicitly as a function of the mean density \(\bar u\), the nonlinear exponent \(m\), and the second-harmonic response ratio
\(\kappa=\widehat K(2k_c)/\widehat K(k_c)\), allowing the critical bifurcation to be classified as supercritical or subcritical.
These results show that the adhesion strength \(a\) determines the onset of instability, whereas the bifurcation type itself is governed primarily by \(\bar u\), \(m\), and \(\kappa\).

Furthermore, we showed that a large nonlinear exponent \(m\) promotes a transition of the critical bifurcation toward subcriticality.
For the tent kernel, the system size \(L\) controls the bifurcation type through the second-harmonic response ratio.
Numerical bifurcation analysis confirmed a supercritical bifurcation for \(L=1.9\) and a subcritical bifurcation with a fold point for \(L=2.1\).

Finally, we examined the mesa limit as \(m\to\infty\) from the perspective of the energy structure.
The local entropy term vanishes under the capacity constraint \(0\le u\le1\), and the limiting energy becomes a capacity-constrained nonlocal attraction energy.
Under suitable assumptions on the kernel, the limiting problem admits a characteristic-function-type minimizer.
This variational structure is consistent with the mesa profiles observed for large \(m\), which consist of saturated regions with \(u\simeq1\) and background regions with \(u\simeq0\).

Although this paper focused primarily on purely attractive kernels representing adhesion, actual cell-cell interactions may switch among adhesion, repulsion, and attraction depending on distance.
Repulsion--attraction-type and adhesion--repulsion-type interactions are also important \cite{CarrilloMurakawaSatoWang2025}.
For such kernels, the sign structures of the Fourier and higher-harmonic responses may change, and it remains to determine how the mode-selection mechanism and bifurcation classification developed here must be modified.
Moreover, on multidimensional domains, several critical modes may emerge simultaneously.
Future work should therefore include the derivation of multimode amplitude equations describing spot, stripe, and network-like patterns, as well as analyses of coarsening, metastability, and the persistence of multiple mesas for large \(m\).

A particularly important direction for future research is the analysis of two-component models.
Although the one-component model provides a natural starting point for understanding the fundamental effects of population pressure, degenerate mobility, and nonlocal interactions, at least two cell populations are required to describe cell-sorting phenomena driven by cell-cell adhesion.
In a two-component system, differences between homotypic and heterotypic cell-cell adhesion can generate a wide variety of patterns, including destabilization of mixed states, interface formation, complete segregation, engulfment, and partial engulfment.
Extending the linear stability analysis, weakly nonlinear analysis, and mesa-limit framework developed in this paper to two-component systems is therefore important from both mathematical and biological perspectives.

\section*{Declaration of competing interest}
The authors have
declared that no competing interests exist.

\section*{CRediT authorship contribution statement}

Shimpei Makida: Methodology, Formal analysis, Investigation,
Validation, Writing -- original draft, Writing -- review \& editing.

Hideki Murakawa: Conceptualization, Methodology, Software,
Formal analysis, Investigation, Validation, Visualization,
Supervision, Funding acquisition, Writing -- original draft,
Writing -- review \& editing.

\section*{Data and code availability}
The source code used in this study is openly available from Zenodo at
\url{https://doi.org/10.5281/zenodo.21768367}.

\section*{Declaration of generative AI and AI-assisted technologies in the manuscript preparation process}

During the preparation of this work, the authors used ChatGPT and Codex (OpenAI), and Claude and Claude Code (Anthropic), to assist with Japanese-to-English translation, language editing, and consistency checking. Codex (OpenAI) was also used to assist in generating the computational code underlying the bifurcation diagrams in Figure 4 and to review all numerical codes employed in this study. All AI-assisted text, code, and review results were carefully examined, edited, tested, and verified by the authors. The authors take full responsibility for the content of the published article.

\section*{Acknowledgment}
This work was partially supported by JSPS KAKENHI Grant Numbers 24H00188 and 21KK0044. 
JSPS had no role in the study design, analysis and interpretation of the results, preparation of the manuscript, or decision to submit the article for publication.

%============================================================
\bibliographystyle{elsarticle-num}
\bibliography{references_physicad}

\appendix
\renewcommand{\thetheorem}{\Alph{section}.\arabic{theorem}}

\section{Fourier responses of representative attractive kernels}
\label{app:kernels}

Although the main text focuses primarily on the tent kernel, for comparison, this appendix summarizes the radial profiles, adhesion-strength functions, Fourier responses, and second-harmonic response ratios of representative purely attractive kernels.
In what follows, we restrict our attention to one spatial dimension and assume that the base kernel \(K_0\) is even, with \(K_0(x)=K_0(r)\), where \(r=|x|\).

For compactly supported kernels, as in the main text, we nondimensionalize the sensing radius and set \(R=1\). Each kernel is normalized so that
\begin{equation}
    \int_{\mathbb R}K_0(x)\,dx=1
    \label{eq:app_kernel_normalization}
\end{equation}
holds. The adhesion-strength function is defined by
\begin{equation}
    \omega(r):=-\frac{d}{dr}K_0(r)
    \qquad (r>0).
    \label{eq:app_omega_def}
\end{equation}
When \(K_0\) decreases with distance, \(\omega(r)\ge0\) represents the strength of adhesion exerted by cells located at distance \(r\). The Fourier response is given by
\begin{equation}
    \widehat K(k)
    =
    \int_{\mathbb R}K_0(x)e^{-ikx}\,dx
    =
    2\int_0^\infty K_0(r)\cos(kr)\,dr.
    \label{eq:app_fourier_def}
\end{equation}
As explained in Subsection~\ref{subsec:model_formulation}, this coincides with the Fourier coefficient of the periodized kernel evaluated at an admissible wavenumber.
The quantity relevant to the weakly nonlinear analysis is the second-harmonic response ratio at the critical wavenumber \(k_c\), defined by
\begin{equation}
    \kappa(k_c)
    :=
    \frac{\widehat K(2k_c)}{\widehat K(k_c)}.
    \label{eq:app_kappa_def}
\end{equation}
In the formulas below, \(\widehat K(0)=1\) is understood by continuous extension. The ratio \(\kappa(k_c)\) is defined whenever \(\widehat K(k_c)\ne0\) and is interpreted by continuous extension at removable singularities.

Table~\ref{tab:kernels_fourier_response} summarizes \(K_0(r)\), \(\omega(r)\), \(\widehat K(k)\), and \(\kappa(k_c)\) for six representative purely attractive kernels.
For the compactly supported kernels---the tent, parabolic, quadratic tent, and raised-cosine kernels---the sensing radius is set to \(R=1\).
The Gaussian and screened Poisson kernels do not have a finite sensing radius; instead, \(\sigma\) and \(\ell\), respectively, represent their interaction lengths.

\begin{table*}[t]
\small
\centering

\caption{
Radial profiles \(K_0(r)\), adhesion-strength functions
\(\omega(r)=-K_0'(r)\), Fourier responses \(\widehat K(k)\), and
second-harmonic response ratios
\(\kappa(k_c)=\widehat K(2k_c)/\widehat K(k_c)\)
for representative one-dimensional purely attractive kernels.
}
\label{tab:kernels_fourier_response}
\renewcommand{\arraystretch}{2.2}
\setlength{\tabcolsep}{5pt}
\resizebox{\linewidth}{!}{%
\begin{tabular}{c|c|c|c|c}
\hline
Name
&
\(K_0(r)\)
&
\(\omega(r)=-K_0'(r)\)
&
\(\widehat K(k)\)
&
\(\kappa(k_c)\)
\\
\hline\hline
Tent
&
\((1-r)_+\)
&
\(\mathbf 1_{(0,1)}(r)\)
&
\(\displaystyle \frac{2(1-\cos k)}{k^2}\)
&
\(\displaystyle \cos^2\!\left(\frac{k_c}{2}\right)\)
\\
\hline
Parabolic
&
\(\displaystyle \frac34\,(1-r^2)_+\)
&
\(\displaystyle \frac32\,r\,\mathbf 1_{(0,1)}(r)\)
&
\(\displaystyle \frac{3(\sin k-k\cos k)}{k^3}\)
&
\(\displaystyle
\frac{\sin(2k_c)-2k_c\cos(2k_c)}
{8\,(\sin k_c-k_c\cos k_c)}
\)
\\
\hline
Quadratic tent
&
\(\displaystyle \frac32\,(1-r)_+^2\)
&
\(\displaystyle 3(1-r)\,\mathbf 1_{(0,1)}(r)\)
&
\(\displaystyle \frac{6(k-\sin k)}{k^3}\)
&
\(\displaystyle
\frac{2k_c-\sin(2k_c)}
{8\,(k_c-\sin k_c)}
\)
\\
\hline
Raised cosine
&
\(\displaystyle \frac12\bigl(1+\cos(\pi r)\bigr)\mathbf 1_{[0,1]}(r)\)
&
\(\displaystyle \frac{\pi}{2}\sin(\pi r)\,\mathbf 1_{(0,1)}(r)\)
&
\(\displaystyle
\frac{\pi^2\sin k}{k\,(\pi^2-k^2)}
\)
&
\(\displaystyle
\cos k_c\,\frac{\pi^2-k_c^2}{\pi^2-4k_c^2}
\)
\\
\hline
Gaussian
&
\(\displaystyle
\frac{1}{\sqrt{\pi}\,\sigma}
\exp\!\left(-\frac{r^2}{\sigma^2}\right)
\)
&
\(\displaystyle
\frac{2r}{\sqrt{\pi}\,\sigma^3}
\exp\!\left(-\frac{r^2}{\sigma^2}\right)
\)
&
\(\displaystyle
\exp\!\left(-\frac{\sigma^2k^2}{4}\right)
\)
&
\(\displaystyle
\exp\!\left(-\frac{3\sigma^2k_c^2}{4}\right)
\)
\\
\hline
Screened Poisson
&
\(\displaystyle
\frac{1}{2\ell}\,e^{-r/\ell}
\)
&
\(\displaystyle
\frac{1}{2\ell^2}\,e^{-r/\ell}
\)
&
\(\displaystyle
\frac{1}{1+\ell^2k^2}
\)
&
\(\displaystyle
\frac{1+\ell^2k_c^2}{1+4\ell^2k_c^2}
\)
\\
\hline
\end{tabular}%
}
\end{table*}

For the tent kernel used primarily in the main text,
\(\omega(r)=\mathbf 1_{(0,1)}(r)\).
Thus, adhesion acts with the same strength independently of distance within the sensing radius \(R=1\), whereas the interaction vanishes outside this range.
In this sense, the tent kernel is the simplest finite-range purely attractive kernel.
By contrast, for the parabolic, quadratic tent, and raised-cosine kernels, \(\omega(r)\) depends on the distance.
Although these kernels have the same finite sensing range, they have different distributions of adhesion strength with respect to distance, and these differences are reflected in their Fourier responses \(\widehat K(k)\) and second-harmonic response ratios \(\kappa(k_c)\).

The Gaussian and screened Poisson kernels do not have a finite sensing radius and instead represent interactions with unbounded support that decay with distance.
For comparison, if \(\sigma=R/3\) is chosen for the Gaussian kernel, then, for \(R=1\),
\[
    \kappa_{\mathrm{gauss}}(k_c)
    =
    \exp\!\left(-\frac{k_c^2}{12}\right).
\]
The screened Poisson kernel is the fundamental solution \(K_{\mathrm{sp}}\) of the elliptic equation
\[
    -d\Delta w+w=u,
    \qquad d=\ell^2,
\]
where \(\ell=\sqrt d\) is the interaction length, so that
\(w=K_{\mathrm{sp}}*u\).
This provides a representative connection between nonlocal adhesion models and elliptic Keller--Segel-type models.
In one dimension, the fundamental solution is
\(K_0(r)=\tfrac{1}{2\ell}e^{-r/\ell}\), as listed in
Table~\ref{tab:kernels_fourier_response}.
It is positive, monotonically decreasing, and satisfies
\eqref{eq:app_kernel_normalization}.
Its Fourier response is the rational function
\(\widehat K(k)=1/(1+\ell^2k^2)\), and its second-harmonic response ratio is
\[
    \kappa_{\mathrm{sp}}(k_c)
    =
    \frac{1+\ell^2k_c^2}{1+4\ell^2k_c^2}.
\]
In particular,
\(\kappa_{\mathrm{sp}}\to1\) as \(\ell k_c\to0\), corresponding to the long-wavelength regime, whereas
\(\kappa_{\mathrm{sp}}\to1/4\) as \(\ell k_c\to\infty\), corresponding to the short-wavelength regime.
Because the Fourier response
\(\widehat K_{\mathrm{sp}}(k)=1/(1+\ell^2k^2)\)
decreases with increasing wavenumber \(k\), high-wavenumber components are more strongly attenuated.
Consequently, when the fundamental wavelength is sufficiently long relative to the interaction length \(\ell\), both the fundamental and second-harmonic modes lie in the low-wavenumber regime, where the difference between their attenuation rates is small, and \(\kappa_{\mathrm{sp}}\) is close to \(1\).
By contrast, when the fundamental wavelength becomes short relative to \(\ell\), the second harmonic is attenuated more strongly than the fundamental mode, and \(\kappa_{\mathrm{sp}}\) decreases toward \(1/4\).

In two dimensions, the fundamental solution can be written in terms of the modified Bessel function \(\mathrm K_0\) as
\(K_{\mathrm{sp}}(x)=\frac{1}{2\pi d}\mathrm K_0(|x|/\ell)\).
Although it has a logarithmic singularity at the origin, its Fourier response retains the same form,
\(\widehat K(k)=1/(1+\ell^2|k|^2)\), and is therefore convenient for both linear stability and weakly nonlinear analyses.

Thus, differences in kernel shape are encoded in the second-harmonic response ratio \(\kappa(k_c)\).
When the domain size is sufficiently large relative to the sensing radius, the critical wavenumber \(k_c\) becomes small.
The fundamental and second-harmonic modes then both lie in the low-wavenumber regime, and \(\kappa\) approaches \(1\) for many kernels, indicating a strong second-harmonic response.
By contrast, as the domain size approaches the sensing radius, the second harmonic enters the attenuation range of the kernel, causing \(\kappa\) to decrease.
Therefore, the bifurcation criterion \eqref{eq:B_def} obtained from the weakly nonlinear analysis in the main text applies in the same form to all the kernels listed in Table~\ref{tab:kernels_fourier_response}.
The system size and kernel shape determine whether the critical bifurcation is supercritical or subcritical through \(\kappa(k_c)\).

\section{Details of the weakly nonlinear analysis}
\label{app:weakly_nonlinear_details}

In this appendix, we present the derivation of the Stuart--Landau equation and the Landau coefficient used in Section~\ref{sec:weakly_nonlinear}. We consider a one-dimensional periodic domain and restrict our attention to the case in which a simple critical mode \(k_c\) is selected.

We first restate the problem setting. Model \eqref{eq:main_model} is given by
\begin{equation}
    \partial_t u
    =
    \partial_{xx}P(u)
    -
    a\nabla \cdot\left(g(u)\partial_x(K*u)\right),
    \label{eq:app_pde}
\end{equation}
where
$
    P(u)=\frac{m}{m+1}u^{m+1},
    \quad
    g(u)=u(1-u^m)
$.
For a parameter \(\varepsilon>0\), we introduce the slow time
\(\tau=\varepsilon^2t\), perturb the adhesion strength near the critical value \(a_c\) as
$
    a=a_c+\varepsilon^2a_2,
$
and expand the solution of the model as
\[
    u=\bar u+\varepsilon u_1(x,\tau)+\varepsilon^2u_2(x,\tau)+\varepsilon^3u_3(x,\tau)+O(\varepsilon^4).
\]
We define the Taylor coefficients of \(P\) and \(g\) by
\[
    D_1=P'(\bar u),\quad
    D_2=\frac12P''(\bar u),\quad
    D_3=\frac16P'''(\bar u),
\]
\[
    G_0=g(\bar u),\quad
    G_1=g'(\bar u),\quad
    G_2=\frac12g''(\bar u).
\]

The left-hand side of model \eqref{eq:main_model} then expands as
\[
    \partial_{t}u=\varepsilon^{3}(\partial_{\tau} u_1)+\cdots.
\]
The degenerate diffusion term is expanded as
\begin{align*}
    \partial_{xx} P(u)
    &=
    \varepsilon D_1 \partial_{xx} u_1
    \\
    &\quad
    +\varepsilon^2
    \left(
        D_1 \partial_{xx} u_2
        +D_2 \partial_{xx}(u_1^2)
    \right)
    \\
    &\quad
    +\varepsilon^3
    \left(
        D_1 \partial_{xx} u_3
        +2D_2 \partial_{xx} \left(u_1 u_2\right)
        +D_3 \partial_{xx} \left(u_1^3\right)
    \right),
\end{align*}
whereas the nonlocal term can be expanded as
\begin{align*}
&-a\nabla \cdot\left(g(u)\partial_x(K*u)\right)
\\
&\quad
=-\varepsilon a_c G_0 \partial_{xx} K\ast u_1
\\
&\quad
-\varepsilon^2 a_c
\left[
    G_0 \partial_{xx} K \ast u_2
    \right.
\left.
    +G_1 \nabla \cdot \left(u_1 \nabla K \ast u_1\right)
\right]
\\
&\quad
-\varepsilon^3
\left[
    a_c G_0 \partial_{xx} K\ast u_3
    \right.
\left.
    +a_c G_1 \nabla \cdot (u_1 \nabla K\ast u_2)
\right.
\\
&\quad\left.
    +a_c G_1 \nabla \cdot (u_2 \nabla K \ast u_1)
\right.
\left.
    +a_c G_2 \nabla \cdot (u_1^2 \nabla K \ast u_1)
    \right.
\\
&\hspace{34mm}\left.
    +a_2 G_0 \partial_{xx} K\ast u_1
\right].
\end{align*}

\subsection*{Terms of order \(\varepsilon\)}

We introduce the operator
$
    \mathcal L_c v
    =
    D_1\partial_{xx}v
    -
    a_cG_0\partial_{xx}(K*v).
    \label{eq:app_Lc}
$
Collecting the terms of order \(\varepsilon\) gives
$
    \mathcal L_cu_1=0.
$
Under the assumption of a simple critical mode, \(u_1\) can be written as
$
    u_1(x,\tau)
    =
    A(\tau)e^{ik_cx}
    +
    \overline{A(\tau)}e^{-ik_cx}.
    \label{eq:app_u1}
$
Substituting this expression into the resulting equation and equating the coefficient of \(e^{ik_cx}\), we obtain the criticality condition
\begin{equation}
    a_cG_0\widehat K(k_c)=D_1.
    \label{eq:app_critical_condition}
\end{equation}
\subsection*{Terms of order \(\varepsilon^2\)}

Collecting the terms of order \(\varepsilon^2\), we obtain
\begin{equation}
    \mathcal L_cu_2
    =
    -
    \left[
        D_2\partial_{xx}(u_1^2)
        -
        a_cG_1\nabla \cdot
        \left(
            u_1\partial_x(K*u_1)
        \right)
    \right].
    \label{eq:app_order2}
\end{equation}
Substituting
\(u_1=Ae^{ik_cx}+\overline A e^{-ik_cx}\), we find that the mean component on the right-hand side vanishes owing to the divergence-form structure, leaving only the second harmonics \(\pm2k_c\). Therefore, \(u_2\) can be written in terms of a real coefficient \(U_2\) as
\begin{equation}
    u_2
    =
    U_2A^2e^{2ik_cx}
    +
    {U_2}\,\overline A^2e^{-2ik_cx}.
    \label{eq:app_u2}
\end{equation}
Substituting this expression for \(u_2\) into \eqref{eq:app_order2} and equating the coefficients of \(e^{2ik_cx}\), we obtain
\[
    -4k_c^2
    \left[
        D_1-a_cG_0\widehat K(2k_c)
    \right]U_2
    =
    2k_c^2
    \left[
        2D_2-a_cG_1\widehat K(k_c)
    \right].
\]
Using the criticality condition \eqref{eq:app_critical_condition} and the second-harmonic response ratio
$
    \kappa=\frac{\widehat K(2k_c)}{\widehat K(k_c)},
$
we obtain
\begin{equation}
    U_2
    =
    \frac{1}{2D_1(\kappa-1)}
    \left(
        2D_2-\frac{D_1G_1}{G_0}
    \right).
    \label{eq:app_U2}
\end{equation}

\subsection*{Terms of order \(\varepsilon^3\)}

At order \(\varepsilon^3\), contributions arise from the slow-time derivative, the perturbation \(a_2\) of the adhesion strength, and the nonlinear interactions between \(u_1\) and \(u_2\). We obtain
\begin{equation}
    \mathcal L_cu_3
    =
    \partial_\tau u_1
    -
    N_3
    +
    a_2G_0\partial_{xx}(K*u_1).
    \label{eq:app_order3}
\end{equation}
Here, the cubic nonlinear term \(N_3\) is given by
\begin{align}
    N_3
    =
    &\ \partial_{xx}
    \left(
        2D_2u_1u_2+D_3u_1^3
    \right)
    \notag\\
    &-
    a_c\nabla \cdot
    \left[
        G_1u_2\partial_x(K*u_1)
    \right.\notag\\
    &\hspace{23mm}\left.
        +G_1u_1\partial_x(K*u_2)
        +G_2u_1^2\partial_x(K*u_1)
    \right].
    \label{eq:app_N3}
\end{align}

The growth rate vanishes for the critical mode. Consequently, when \(u_3\) in \eqref{eq:app_order3} is expanded as a Fourier series, the coefficient of \(e^{ik_cx}\) on the left-hand side vanishes. The coefficient of \(e^{ik_cx}\) on the right-hand side must therefore also vanish.

First, the \(e^{ik_cx}\) component of the adhesion-strength perturbation term
$
    a_2G_0\partial_{xx}(K*u_1)
$
is
\[
    -a_2G_0\widehat K(k_c)k_c^2A.
\]
We next calculate the \(e^{ik_cx}\) component of \(-N_3\). The contribution from the diffusion term is
\[
    k_c^2
    \left(
        2D_2U_2+3D_3
    \right)
    |A|^2A,
\]
whereas the contribution from the nonlocal adhesion term is
\[
    -k_c^2\frac{D_1}{G_0}
    \left[
        G_1U_2(2\kappa-1)+G_2
    \right]
    |A|^2A.
\]
Therefore, the fundamental-mode component of \(-N_3\) is
\[
    -k_c^2
    \left[
        -2D_2U_2
        -3D_3
        +
        \frac{D_1G_1}{G_0}U_2(2\kappa-1)
        +
        \frac{D_1G_2}{G_0}
    \right]
    |A|^2A.
\]

Combining these results, we obtain the Stuart--Landau equation
$
    \frac{dA}{d\tau}=\sigma A-c_{\mathrm L}|A|^2A,
$
where \(\sigma\) and \(c_{\mathrm L}\) are defined by
\begin{equation}
    \sigma=a_2G_0\widehat K(k_c)k_c^2,
\end{equation}
and
\begin{equation}
    c_{\mathrm L}
    =
    k_c^2
    \left[
        U_2
        \left(
            2D_2
            -
            \frac{D_1G_1}{G_0}(2\kappa-1)
        \right)
        +
        3D_3
        -
        \frac{D_1G_2}{G_0}
    \right].
    \label{eq:app_landau}
\end{equation}
This agrees with \eqref{eq:landau_constant} in the main text.
\section{Numerical method}
\label{app:numerical_scheme}

This appendix describes the numerical method used for the time-dependent simulations in
Sections~\ref{sec:linear_stability} and~\ref{sec:mesa_limit}.
In model~\eqref{eq:main_model}, the nonlinear diffusion term becomes highly stiff
for large values of the nonlinear exponent \(m\). An explicit treatment of this term would require extremely small time steps.

We apply the idea underlying Murakawa's efficient linear scheme for nonlinear diffusion problems
\cite{MurakawaEfficientLS} to equation~\eqref{eq:main_model}, which contains a nonlocal adhesion term.
In this approach, the nonlinear diffusion is treated linearly and implicitly in terms of an auxiliary variable,
thereby reducing each time step to the solution of a single linear equation.
The nonlocal adhesion field is evaluated explicitly from the density at the previous time step, whereas the adhesive flux is treated implicitly using the auxiliary variable. This makes it possible to perform fast and stable time-dependent simulations
even for large \(m\), while avoiding nonlinear implicit solvers based on Newton's method.

\subsection{Time-discrete scheme}
\label{app_subsec:time_scheme}

Model \eqref{eq:main_model} can be rewritten as
\begin{equation}
    \frac{\partial u}{\partial t}
    =
    \Delta \beta(u)
    -
    a\nabla \cdot
    \left[
        \left(
            u-\frac{m+1}{m}\beta(u)
        \right)
        \nabla(K*u)
    \right],
    \label{eq:app_pde_beta_form}
\end{equation}
where
$
    \beta(u)
    :=
    \frac{m}{m+1}u^{m+1}$.
The basic idea of the linear scheme \cite{MurakawaEfficientLS} is to use the relation
$
    \frac{\partial u}{\partial t}
    =
    \frac{1}{\beta'(u)}\frac{\partial }{\partial t}\beta(u)
$
and introduce an auxiliary variable $w$ representing \(\beta(u)\).

Let \(u^n\) denote the time-discrete approximate solution at time \(t^n\).
\begin{equation}
\mu^n = \frac{1}{\beta'(u^n) + \varepsilon_\mu}.
\label{eq:mu_equation}
\end{equation}
\begin{align}
\mu^n \frac{w^{n+1} - \beta(u^n)}{\Delta t}
&= \Delta w^{n+1} \notag \\
&\hspace{-7mm} - a \nabla \cdot \biggl[
\left( u^{n+1} - \frac{m+1}{m} w^{n+1} \right)
\nabla(K * u^n) \biggr]. \label{eq:w_ls}
\end{align}
\begin{equation}
u^{n+1} - u^n
= \mu^n \left( w^{n+1} - \beta(u^n) \right).
\label{eq:u_ls}
\end{equation}
Here, \(\varepsilon_\mu>0\) is a small positive constant introduced to avoid division by zero due to \(\beta'(u)=0\) at \(u=0\).
In the implementation used in this study, we set
$
    \varepsilon_\mu=10^{-10}
$.
Substituting equation \eqref{eq:u_ls} into equation \eqref{eq:w_ls} yields a closed linear elliptic equation for $w^{n+1}$ and an update formula for $u^{n+1}$ that involves only direct substitution. Thus, we obtain the following linear scheme:
\begin{align}
\mu^n =& 1 / (\beta'(u^n) + \varepsilon_\mu), \label{eq:app_density_update} \\
\mu^n w^{n+1} &- \Delta t \Delta w^{n+1} + a \Delta t \nabla \cdot \left[ \left(\mu^n - \frac{m+1}{m}\right)w^{n+1} \nabla(K * u^n) \right] \notag \\
=& \mu^n \beta(u^n) - a \Delta t \nabla \cdot \left[ (u^n - \mu^n\beta(u^n)) \nabla(K * u^n) \right], \label{eq:app_time_scheme} \\
u^{n+1}  =&  u^n + \mu^n (w^{n+1} - \beta(u^n)). \label{eq:u_equation}
\end{align}

\subsection{Fully discrete scheme with spatial discretization}
\label{app_subsec:full_scheme}

\noeqref{eq:app_time_scheme}
Consider the one-dimensional periodic domain \(\Omega=[-L,L)\).
We discretize the time-discrete scheme \eqref{eq:app_density_update}--\eqref{eq:u_equation}
in space using a conservative finite-volume method.
Partitioning the domain uniformly into \(N_X\) cells, the spatial mesh size is $    \Delta x=\frac{2L}{N_X}$, and
\[
    x_i=-L+\left(i+\frac12\right)\Delta x,
    \qquad
    i=0,\ldots,N_X-1
\]
denotes the center of the \(i\)-th cell. Indices are interpreted periodically in accordance with the periodic boundary conditions.

Let \(u_i^n\) denote the value at the \(i\)-th cell center and \(w_i^{n+1}\) the corresponding auxiliary unknown. Set
\[
    \beta_i^n:=\beta(u_i^n),
    \qquad
    \mu_i^n:=\frac{1}{\beta'(u_i^n)+\varepsilon_\mu},
    \qquad
    c_m:=\frac{m+1}{m}
\]
and define
\[
    A_i^n:=\mu_i^n-c_m,
    \qquad
    B_i^n:=u_i^n-\mu_i^n\beta_i^n.
\]

The convolution \(K*u^n\) is evaluated as a periodic convolution using the FFT
(fast Fourier transform).
Rather than sampling the kernel in physical space, we use the analytical Fourier response
\(\widehat K(k)\) given in \ref{app:kernels}.

More specifically, denoting the discrete Fourier transform by \(\mathcal F_h\), let
\[
    \widehat u^n(k_\ell)
    :=
    \mathcal F_h[u^n](k_\ell),
\]
where the discrete wavenumbers are
\[
    k_\ell=\frac{\pi}{L}\ell,
    \qquad
    \ell=-\frac{N_X}{2},\ldots,\frac{N_X}{2}-1.
\]
For each discrete wavenumber \(k_\ell\), we set
\[
    \widehat{(K*u^n)}(k_\ell)
    =
    \widehat K(k_\ell)\widehat u^n(k_\ell)
\]
and obtain the values in physical space by the inverse discrete Fourier transform:
\[
    K*u^n
    =
    \mathcal F_h^{-1}
    \left[
        \widehat K\widehat u^n
    \right].
\]
Thus, the Fourier response used in the numerical computations coincides with
\(\widehat K\) used in the linear stability and weakly nonlinear analyses.

We define the adhesion velocity at the cell interface \(x_{i+1/2}\) by
\begin{equation}
    V_{i+1/2}^n
    =
    a\frac{(K*u^n)_{i+1}-(K*u^n)_i}{\Delta x}.
    \label{eq:app_linear_velocity}
\end{equation}
Based on the direction of this velocity, the upwind value is defined by
\[
    X_{\mathrm{up},i+1/2}
    =
    \begin{cases}
        X_i,
        & V_{i+1/2}^n\ge0,\\
        X_{i+1},
        & V_{i+1/2}^n<0,
    \end{cases}
\]
where \(X\) denotes
\(A^n\), \(B^n\), or \(w^{n+1}\).

The discrete flux for the auxiliary variable \(w^{n+1}\) is defined by
\begin{equation}
    J_{i+1/2}^{n+1}
    =
    -
    \frac{w_{i+1}^{n+1}-w_i^{n+1}}{\Delta x}
    +
    V_{i+1/2}^n
    \left(
        A_{\mathrm{up},i+1/2}^n w_{\mathrm{up},i+1/2}^{n+1}
        +
        B_{\mathrm{up},i+1/2}^n
    \right).
    \label{eq:app_linear_flux}
\end{equation}
The finite-volume discretization of \eqref{eq:app_time_scheme} is then given by
\begin{equation}
    \mu_i^n
    \frac{w_i^{n+1}-\beta_i^n}{\Delta t}
    +
    \frac{J_{i+1/2}^{n+1}-J_{i-1/2}^{n+1}}{\Delta x}
    =
    0,
    \qquad
    i=0,\ldots,N_X-1.
    \label{eq:app_fully_discrete}
\end{equation}
This is a linear equation for \(w^{n+1}=(w_0^{n+1},\ldots,w_{N_X-1}^{n+1})\).
Owing to the periodic boundary conditions, the resulting linear system is cyclic tridiagonal.
In the implementation used in this study, this linear system is solved using a direct solver for cyclic tridiagonal matrices.
The linear system itself can be solved with a computational cost of \(O(N_X)\).
However, because the convolution is evaluated using the FFT at each time step,
the computational cost per step is dominated by the \(O(N_X\log N_X)\) cost of the convolution.

After solving the linear system, the density is updated according to
\begin{equation}
    u_i^{n+1}
    =
    u_i^n+\mu_i^n(w_i^{n+1}-\beta_i^n).
    \label{eq:app_density_update_discrete}
\end{equation}
Summing \eqref{eq:app_fully_discrete} over all cells, the flux differences cancel by periodicity. Hence,
\[
    \Delta x\sum_i u_i^{n+1}
    =
    \Delta x\sum_i u_i^n
\]
holds up to roundoff error. Thus, the discrete mass is conserved by construction.

In the computations presented in this study, we use a fixed time step $\Delta t=10^{-4}$ and a fixed number of spatial cells $N_X=512$.

\subsection{Initial condition}
\label{app_subsec:numerical_initial}

In the numerical simulations, we use no random numbers for the initial perturbation,
but instead employ a deterministic initial condition that is reproducible in all computations.
On the one-dimensional periodic domain \(\Omega=[-L,L)\), with mean density \(\bar u\), the initial condition is prescribed by
\begin{equation}
    u_0(x)
    =
    \bar u
    \left[
        1+
        \eta\,\frac{1}{N_{\max}}
        \sum_{n=1}^{N_{\max}}
        \cos\left(\frac{n\pi x}{L}\right)
    \right].
    \label{eq:numerical_initial_data}
\end{equation}
Here, \(\eta>0\) is a parameter representing the magnitude of the perturbation,
and we set \(\eta=0.05\) in the numerical simulations presented in this study.

The number of spatial cells \(N_X\) is taken to be even, and the maximum mode number is set to
$    N_{\max}=\frac{N_X}{2}-1
$.
The cosine modes representable on a grid of \(N_X\) points are
\(n=1,\ldots,N_X/2\), and the highest mode \(n=N_X/2\) is an alternating mode
with wavelength \(2\Delta x\), whose sign changes between adjacent grid points.
Because this mode is susceptible to numerical high-frequency noise and
aliasing errors, it is excluded from the initial perturbation.

Thus, the initial condition \eqref{eq:numerical_initial_data} contains all resolvable
cosine modes, from low to high frequencies, with equal amplitude coefficients.
Rather than artificially selecting a particular mode, this initial condition is used
to determine whether the mode with the largest growth rate predicted by linear stability
analysis naturally becomes dominant during the early stage of the time evolution.
Moreover, because no random numbers are used, variations arising from differences in
the initial perturbation can be eliminated when comparing different parameter values.

\subsection{Context of the present numerical scheme}
\label{app_subsec:numerical_position}

The numerical method described above applies Murakawa's linear scheme \cite{MurakawaEfficientLS} to
the nonlocal cell adhesion equation \eqref{eq:app_pde_beta_form}.
Its main features can be summarized as follows.
First, because the nonlinear diffusion is treated linearly and implicitly in terms of the auxiliary variable \(w\approx\beta(u)\),
the severe explicit time-step restriction arising from diffusion can be avoided.
Second, each time step requires only a single solution of a cyclic tridiagonal linear system, making the scheme straightforward to implement.
Third, because the scheme is based on a conservative finite-volume discretization, the discrete mass is conserved by construction.
Fourth, by locally choosing the stabilization parameter as \(\mu^n\approx1/\beta'(u^n)\),
the nonlinear diffusion can be approximated efficiently while retaining linearity.
This method enables the formation of sharp mesa-type patterns for large \(m\) to be tracked stably and efficiently.

\section{Construction of the steady-state bifurcation diagram}
\label{app:numerical_bifurcation}

Here, we describe the construction of the steady-state bifurcation diagram shown in Fig.~\ref{fig:tent_system_bifurcation}. The objective is not to perform direct time-dependent simulations, but rather to numerically continue the nonhomogeneous steady-state branch bifurcating from the homogeneous steady state and to verify the bifurcation type determined by weakly nonlinear analysis. In what follows, we set \(m=1\), \(\bar u=1/2\), and \(R=1\), use the tent kernel, and consider the case in which the fundamental mode \(k_c=\pi/L\) is the critical mode.

\subsection{Formulation of the steady-state problem}
\label{app_subsec:bif_formulation}

It follows from the gradient-flow formulation \eqref{eq:gradient_flow_form} that, for any smooth periodic steady-state solution satisfying \(0<u<1\), the chemical potential
\[
    \mu(u,a)=-\log(1-u^m)-aK*u
\]
is spatially constant. Therefore, in the steady-state bifurcation analysis, rather than considering the time-evolution residual, we solve the condition
\begin{equation}
    \mu(u,a)-\frac{1}{|\Omega|}\int_\Omega\mu(u,a)\,dx=0
    \label{eq:app_bif_steady}
\end{equation}
that eliminates the nonconstant component of the chemical potential.

\subsection{Even-symmetric cosine expansion and Galerkin residual}
\label{app_subsec:bif_galerkin}

To track the even-symmetric solution branch bifurcating from the fundamental mode, we approximate the solution by the cosine expansion
\begin{equation}
    u(x)=\bar u+\sum_{n=1}^{N_c}b_n\cos\!\left(\frac{n\pi x}{L}\right).
    \label{eq:app_cosine_expansion}
\end{equation}
This representation automatically fixes the mean density \(\bar u\), while restriction to the even-symmetric subspace eliminates the need for a phase condition associated with periodic translational symmetry.

We evaluate \eqref{eq:app_cosine_expansion} at uniformly spaced grid points \(x_j\) (\(j=0,\ldots,N_x-1\)) and compute the convolution \(K*u\) using the FFT. For continuation of the steady-state branch, we employ a cosine collocation method suited to the treatment of smooth solution branches by Newton's method; this differs from the piecewise-constant finite-volume discretization used for the time evolution in \ref{app:numerical_scheme}. The expansion coefficients \(b\) are the unknowns, and the grid points serve as the collocation points. Rather than sampling the kernel in physical space, we use the analytical Fourier response of the tent kernel, \(\widehat K(k)=2(1-\cos(kR))/(kR)^2\) (\(\widehat K(0)=1\)). At each grid point, we compute
\[
    \mu_j=-\log\!\bigl(1-u(x_j)^m\bigr)-a\,(K*u)(x_j),
    \qquad
    r_j=\mu_j-\frac{1}{N_x}\sum_{\ell=0}^{N_x-1}\mu_\ell
\]
and define the residual as the cosine components of \(r\):
\begin{equation}
    F_n(b,a)
    =
    \frac{2}{N_x}\sum_{j=0}^{N_x-1}
    r_j\cos\!\left(\frac{n\pi x_j}{L}\right),
    \qquad
    n=1,\ldots,N_c.
    \label{eq:app_galerkin_residual}
\end{equation}
Steady-state solutions are obtained as coefficient vectors \(b=(b_1,\ldots,b_{N_c})\) satisfying \(F(b,a)=0\). In the numerical computations, we use \(N_x=1024\) as the baseline resolution and set the number of modes to \(N_c=12\) for the supercritical case \(L=1.9\) and \(N_c=16\) for the subcritical case \(L=2.1\). We increase \(N_x\) and \(N_c\) as necessary and verify that the shape of the solution branch remains unchanged.

\subsection{Construction of seed solutions}
\label{app_subsec:bif_seed}

As starting points for continuation, we fix the fundamental-mode coefficient \(b_1=A\) at a small value and construct small-amplitude seed solutions by solving \(F(b,a)=0\) for the remaining unknowns \((a,b_2,\ldots,b_{N_c})\) using Newton's method. As the initial guess for Newton's method, we assign to \(a\) the value \(a_c+\alpha A^2\) obtained from the approximate relation
\begin{equation}
    a-a_c\simeq\alpha A^2,
    \qquad
    \alpha=\frac{c_{\mathrm L}}{4\,G_0\,\widehat K(k_c)\,k_c^2},
    \qquad
    G_0=\bar u(1-\bar u^m)
    \label{eq:app_seed_guess}
\end{equation}
derived from the Stuart--Landau equation, while setting the initial guesses for the higher-order coefficients \(b_2,\ldots,b_{N_c}\) to \(0\). This procedure is based on the weakly nonlinear picture that, near the critical point, the amplitude \(A\) is small, the solution is approximated by the fundamental mode \(b_1\), and \(a\) is displaced from its critical value by \(\alpha A^2\). We generate a sequence of seed solutions by increasing \(A\) incrementally from \(10^{-4}\) to \(5\times10^{-2}\) and using the solution obtained for each \(A\), consisting of \(a\) and \(b_2,\ldots,b_{N_c}\), as the initial guess for Newton's method at the next value of \(A\).

\subsection{Continuation procedure}
\label{app_subsec:bif_continuation}

Starting from the seed solutions obtained above, we track the solution branch in the unknowns \((b,a)\). In the supercritical case (\(L=1.9\)), we first attempt standard pseudo-arclength continuation (\texttt{PALC} in \texttt{BifurcationKit.jl}) using the adhesion strength \(a\) as the continuation parameter. Because the solution branch is initialized in the immediate vicinity of the pitchfork bifurcation point, the bordered linear system that directly uses \(a\) as the continuation parameter may become singular. In that event, we switch to a formulation of the same steady-state residual in which the amplitude \(A=b_1\) is treated as the parameter and \(a\) as an unknown. If both approaches fail, we use direct pseudo-arclength continuation of the same residual as a final fallback.

In the subcritical case (\(L=2.1\)), the branch undergoes a fold and turns from the \(a<a_c\) side toward the large-amplitude regime; hence, we employ pseudo-arclength continuation in the unknowns \(z=(b_1,\ldots,b_{N_c},a)\) from the outset. More precisely, an arclength condition based on the tangent direction of the known branch is appended to \(F(b,a)=0\), allowing the branch to be tracked through the fold. The arclength step is gradually increased after each accepted step; if the residual exceeds the prescribed tolerance or the solution violates the physical bounds, the step is reduced and the computation is repeated. Fold candidates are detected as points at which the discrete trend of \(a\) in the continuation data reverses direction.

\subsection{Stability classification}
\label{app_subsec:bif_stability}

The stability of the homogeneous branch is determined from the linear dispersion relation: it is stable for \(a<a_c\) and unstable for \(a>a_c\) with respect to perturbations at fixed mean density.
The stability of the nonhomogeneous branch is classified by linearizing the reduced system
\[
    \frac{db}{dt}=G(b,a)
\]
obtained by projecting the time-evolution equation \(\partial_t u=\partial_x(g(u)\partial_x\mu)\) onto the space of even-symmetric cosine coefficients. We compute the Jacobian \(D_bG(b,a)\) by finite differences and determine stability from the sign of the largest real part of its eigenvalues. A negative largest real part indicates stability, whereas a positive value indicates instability. Points at which the absolute value of the largest real part does not exceed the threshold \(10^{-7}\) are treated as neutral and excluded from the plotted stable and unstable branches. Thus, the solid and dashed portions of the nonhomogeneous branches in Fig.~\ref{fig:tent_system_bifurcation} indicate stability and instability, respectively, within this even-symmetric subspace.

\subsection{Plotting conventions}
\label{app_subsec:bif_plot}

In the bifurcation diagrams, the horizontal axis represents the adhesion strength \(a\), and the vertical axis represents the magnitude of nonhomogeneity, \(\max_x u(x)-\min_x u(x)\). The homogeneous branch is represented by the horizontal line at ordinate \(0\), with stable branches drawn as solid curves and unstable branches as dashed curves. In the subcritical case, the fold point at which the branch turns is indicated by an open circle.

To plot only physically meaningful points, we retain only bifurcation points satisfying \(\min_x u\ge-10^{-10}\), \(\max_x u\le1-10^{-4}\), and \(\max_x u-\min_x u\le1+10^{-8}\). The margin \(10^{-4}\) below the upper bound is introduced to avoid the logarithmic singularity at \(u=1\). The numerical solution itself is not clipped, and candidates that violate the physical bounds are not accepted as bifurcation points.

\subsection{AI-assisted code development}
Codex (OpenAI) was used to assist in generating the computational code for the numerical continuation and stability calculations underlying the bifurcation diagrams in Figure 4. The mathematical formulation and numerical methodology were specified by the authors. The AI-assisted code was subsequently reviewed, modified, tested, and validated by the authors.

\section{Proofs of the Mesa-Limit Results}
\label{app:mesa_limit_proofs}

 In this appendix, we prove the results on the mesa limit stated in Section~\ref{sec:mesa_limit}.

In what follows, to make the dependence on \(m\) explicit, we write \(F_m\) and \(E_m\), respectively, for the quantities already defined in \eqref{eq:F_def} and \eqref{eq:energy}.

\subsection{Limit of the Local Entropy Term}
\label{app_subsec:local_entropy_limit}

We first establish a basic estimate for the local entropy term.
It follows from \eqref{eq:F_def} that, for \(0\le u\le1\),
\[
    0\le F_m(u)\le F_m(1).
\]
Moreover, for
$
    F_m(1)
    =
    \int_0^1-\log(1-s^m)\,ds
$
the series expansion
$
    -\log(1-s^m)
    =
    \sum_{\ell=1}^{\infty}\frac{s^{m\ell}}{\ell}
$
gives
\[
    F_m(1)
    =
    \sum_{\ell=1}^{\infty}
    \frac{1}{\ell(m\ell+1)}
    \le
    \frac{\pi^2}{6m}.
\]
Therefore,
\begin{equation}\label{eq:local_entropy_uniform}
    0\le F_m(u)
    \le
    F_m(1)
    \to 0
    \qquad
    (m\to\infty).
\end{equation}
This estimate yields \eqref{eq:vanishing_entropy_term}.

\subsection{Proof of \(\Gamma\)-Convergence}
\label{app_subsec:gamma_convergence_proof}

We prove Theorem~\ref{thm:mesa_energy_limit}. We first recall the definition of $\Gamma$-convergence.
\begin{definition}
The functional $E_m:L^\infty(\Omega) \cap \mathcal{P}_M(\Omega) \to \mathbb{R}\cup{\{+\infty\}}$ is said to $\Gamma$-converge to the functional $E_{\infty}:L^\infty(\Omega) \cap \mathcal{P}_M(\Omega) \to \mathbb{R}\cup\{+\infty\}$ with respect to the weak-\(*\) \(L^\infty(\Omega)\) topology if the following conditions hold.
\begin{itemize}
    \item ($\Gamma$-liminf inequality) For any $u \in L^\infty(\Omega) \cap \mathcal{P}_M(\Omega)$ and any sequence $u_m$ converging to $u$ with respect to the weak-\(*\) \(L^\infty(\Omega)\) topology, the following inequality holds:
\begin{equation*}
    E_{\infty}[u]\le \liminf_{m \to \infty} E_{m}[u_m].
\end{equation*}
\item (Existence of a recovery sequence) For any $u \in L^\infty(\Omega) \cap \mathcal{P}_M(\Omega)$, there exists a sequence $u_m$ converging to $u$ with respect to the weak-\(*\) \(L^\infty(\Omega)\) topology such that
\begin{equation*}
    E_{\infty}[u]=\lim_{m \to \infty} E_{m}[u_m].
\end{equation*}
\end{itemize}
\end{definition}

The $\Gamma$-convergence of the energy functionals $E_m$ can be established by verifying the preceding definition directly.
\begin{proof}[Proof of Theorem~\ref{thm:mesa_energy_limit}]
The estimate \eqref{eq:uniform_energy_bound} follows immediately from \eqref{eq:vanishing_entropy_term}.

We now verify the \(\Gamma\)-convergence using \eqref{eq:uniform_energy_bound}.
The set \(\mathcal A_M\) is a bounded weak-\(*\) closed subset of \(L^\infty(\Omega)\). Since \(L^1(\Omega)\) is separable, the Banach--Alaoglu theorem implies that \(\mathcal A_M\) is sequentially compact in the weak-\(*\) \(L^\infty(\Omega)\) topology.
Moreover, if \(u_n\rightharpoonup^\ast u\) in \(L^\infty(\Omega)\) as $n \to \infty$, then, in particular, \(u_n\rightharpoonup u\) weakly in \(L^2(\Omega)\).

Under the assumptions, the convolution operator \(T_K:u\mapsto K*u\) is compact on \(L^2(\Omega)\).
Indeed, since $T_K u \in L^2(\Omega)$, it admits a Fourier series expansion.
Finite sums then yield compact approximating operators, which converge to $T_K$ in the operator norm by Plancherel's theorem and the Riemann--Lebesgue lemma.
Since the space of compact operators is closed in the operator norm, it follows that $T_K$ itself is compact.
Consequently, if \(u_n\rightharpoonup u\) weakly in \(L^2(\Omega)\) as $n \to \infty$, then
\[
    T_Ku_n\to T_Ku
    \qquad
    \text{strongly in }L^2(\Omega).
\]
It follows that the quadratic form
$
    Q[u]:=\int_\Omega u(K*u)\,dx
    =:
    \langle u,T_Ku\rangle
$
is continuous with respect to weak convergence in \(L^2(\Omega)\). Indeed,
\[
    Q[u_n]-Q[u]
    =
    \langle u_n-u,T_Ku\rangle
    +
    \langle u_n,T_K(u_n-u)\rangle.
\]
The first term converges to \(0\) because \(u_n-u\rightharpoonup0\) and \(T_Ku\in L^2(\Omega)\).
The second term converges to \(0\) because \(T_K(u_n-u)\to0\) strongly in \(L^2\) and \(\{u_n\}\) is bounded in \(L^2\).
Therefore,
\begin{equation}
    Q[u_n]\to Q[u]\qquad \text{as $n \to \infty$}.
    \label{eq:interaction_weak_continuity}
\end{equation}

We now prove the \(\Gamma\)-liminf inequality using \eqref{eq:interaction_weak_continuity} and \eqref{eq:uniform_energy_bound}.
Consider a sequence satisfying
\[
    u_m\rightharpoonup^\ast u
    \qquad
    \text{in }L^\infty(\Omega).
\]
First, suppose that
$
    \liminf_{m\to\infty} E_m[u_m]=+\infty.
$
In this case, the \(\Gamma\)-liminf inequality is immediate.
We therefore consider the alternative
$
    \liminf_{m\to\infty} E_m[u_m]<+\infty.
$
Choose a subsequence \(u_m\) along which \(\liminf_{m\to\infty} E_m[u_m]\) is attained, not relabeled.
We may then assume that \(u_m\in\mathcal A_M\).
Moreover, the weak-\(*\) closedness of \(\mathcal A_M\) implies that \(u\in\mathcal A_M\).
By \eqref{eq:uniform_energy_bound},
\[
    \liminf_{m\to\infty} E_m[u_m]
    \ge
    \liminf_{m\to\infty}
    \left(
        -\frac a2
        \int_\Omega u_m(K*u_m)\,dx
    \right).
\]
By \eqref{eq:interaction_weak_continuity}, the right-hand side equals
$
    -\frac a2
    \int_\Omega u(K*u)\,dx
    =
    E_\infty[u]
$, and hence
\begin{equation}
    \liminf_{m\to\infty}E_m[u_m]
    \ge
    E_\infty[u]
    \label{eq:gamma_liminf}
\end{equation}
follows.

We next construct a recovery sequence.
Let \(u\in\mathcal A_M\). It suffices to take \(u_m=u\).
Then \eqref{eq:uniform_energy_bound} gives
\[
    0\le E_m[u_m]-E_\infty[u]\to0
    \quad\text{as }m\to\infty.
\]
Now let $u \notin \mathcal{A}_{M}$.
Since $E_{m}[u]=\infty$ in this case, the constant sequence $u_m=u$ is clearly a recovery sequence.

We conclude that \(E_m\) \(\Gamma\)-converges to \(E_\infty\).
This completes the proof of Theorem~\ref{thm:mesa_energy_limit}.
\end{proof}

\subsection{Existence and convergence of minimizers of $E_m$}
\label{app_subsec:minimizer_convergence_proof}
We prove Proposition~\ref{prop:minimizer_for_m} by the direct method.
\begin{proof}[Proof of Proposition~\ref{prop:minimizer_for_m}]
Fix $m\ge 1$. First, since $0<M\le |\Omega|$, the function
$u\equiv M/|\Omega|$ belongs to $\mathcal A_M$, and hence $\mathcal A_M$ is nonempty.
Moreover, the Banach--Alaoglu theorem implies that $\mathcal A_M$ is sequentially compact in the weak-\(*\) \(L^\infty(\Omega)\) topology.

We first verify that $E_m$ is bounded from below on $\mathcal A_M$.
For $u\in \mathcal A_M$, we have $F_m(u)\ge 0$. Moreover, Young's inequality applied to
$
Q[w]=\int_\Omega w(K*w)\,dx
$
gives
\[
\begin{aligned}
|Q[u]|
&=
\left|\int_\Omega u(K*u)\,dx\right| \\
&\le
\|u\|_{L^2}\|K*u\|_{L^2} \\
&\le
\|K\|_{L^1}\|u\|_{L^2}^2
\le \|K\|_{L^1}M
\end{aligned}
\]
and therefore
\[
E_m[u]
=
\int_\Omega F_m(u)\,dx-\frac a2 Q[u]
\ge
-\frac a2\|K\|_{L^1}M.
\]

We next establish the sequential weak lower semicontinuity of $E_m$ on $L^2(\Omega)$.
We begin with the local term. The function $F_m$ is convex and lower semicontinuous on $[0,1]$.
Since it is defined to be $+\infty$ outside $[0,1]$,
$F_m$ is convex and lower semicontinuous on $\mathbb R$.
The functional $\int_{\Omega} F_m(u)\,dx$ therefore admits a representation in terms of the Legendre conjugate on $L^2(\Omega)$ \cite[Corollary]{rockafellar1968},
which shows that the local term is lower semicontinuous with respect to weak convergence in $L^2(\Omega)$.
For the nonlocal term
$
Q[w]
$
we note that, since $K\in L^1(\Omega)$, the periodic convolution operator
$
T_Kw=K*w
$
is compact on $L^2(\Omega)$.
Consequently, if $v_n\rightharpoonup v$ weakly in $L^2(\Omega)$, then
$
Q[v_n]\to Q[v].
$
Thus, $E_m$ is weakly lower semicontinuous.

Let $\{u_n\}\subset \mathcal A_M$ be a minimizing sequence for $E_m$; that is,
$
\lim_{n \to \infty} E_m[u_n]= \inf_{u \in \mathcal A_{M}} E_m[u]
$
holds.
By the sequential compactness of $\mathcal A_M$ in the weak-\(*\) \(L^\infty(\Omega)\) topology, after passing to a subsequence, we may assume that
$
u_n\stackrel{*}{\rightharpoonup}u
\quad\text{in }L^\infty(\Omega).
$
Since $\Omega$ has finite measure, $L^2(\Omega)\subset L^1(\Omega)$.
It follows that
$
u_n\rightharpoonup u
\quad\text{weakly in }L^2(\Omega).
$
Moreover, $u \in \mathcal A_{M}$.
The weak lower semicontinuity of $E_{m}$ therefore yields
\[
\begin{aligned}
E_m[u]
&\le
\liminf_{n\to\infty}E_m[u_n]
=
\inf_{u \in \mathcal A_{M}}E_m[u].
\end{aligned}
\]
Hence, $u$ is a minimizer of $E_m$ over $\mathcal A_M$.
\end{proof}
We prove Corollary~\ref{cor:mesa_minimizer_convergence} directly from the definition.
\begin{proof}[Proof of Corollary~\ref{cor:mesa_minimizer_convergence}]
For each \(m\), let \(u_m\in\mathcal A_M\) be a minimizer of \(E_m\).
Since \(\mathcal A_M\) is sequentially compact in the weak-\(*\) \(L^\infty(\Omega)\) topology, there exist a subsequence \(m_j\to\infty\) and
\(u_\infty\in\mathcal A_M\) such that
\[
    u_{m_j}\rightharpoonup^\ast u_\infty
    \quad
    \text{in } L^\infty(\Omega)
    \quad \text{as } j \to \infty.
\]

For arbitrary \(v\in\mathcal A_M\), let $v_m$ be a recovery sequence.
By the minimality of \(u_m\),
\[
    E_m[u_m]\le E_m[v_m].
\]
Moreover, the defining property of a recovery sequence gives
\[
    E_m[v_m]\to E_\infty[v].
\]
On the other hand, the \(\Gamma\)-liminf inequality yields
\[
    E_\infty[u_\infty]
    \le
    \liminf_{j\to\infty}E_{m_j}[u_{m_j}].
\]
Therefore,
\[
\begin{aligned}
    E_\infty[u_\infty]
    &\le
    \liminf_{j\to\infty}E_{m_j}[u_{m_j}] \\
    &\le
    \limsup_{j\to\infty}E_{m_j}[u_{m_j}] \\
    &\le
    \lim_{j\to\infty}E_{m_j}[v_{m_j}]
    = E_\infty[v].
\end{aligned}
\]
Since \(v\in\mathcal A_M\) was arbitrary, \(u_\infty\) is a minimizer of \(E_\infty\).
This proves Corollary~\ref{cor:mesa_minimizer_convergence}.
\end{proof}

\subsection{Proof of the existence of a mesa-type minimizer}
\label{app_subsec:characteristic_mesa_proof}
To prove Proposition~\ref{prop:characteristic_mesa}, we use the Bauer maximum principle \cite[B.3.4 Corollary]{niculescupersson}.

We first recall the definition of an extreme point.
\begin{definition}
Let $E$ be a real vector space, and let $K \subset E$ be a convex set.
A point $z \in K$ is called an extreme point of $K$ if it satisfies the following condition:
for $x,y \in K$ and $\lambda \in (0,1)$, the identity
\[
  z=(1-\lambda)x+\lambda y
\]
necessarily implies $x=y=z$.
\end{definition}

The Bauer maximum principle states that a convex function attains its maximum over a compact convex set at an extreme point.
\begin{proposition}[Bauer maximum principle]
Let $K$ be a nonempty compact convex subset of a locally convex Hausdorff space.
Let $f:K\to \mathbb{R}$ be an upper semicontinuous convex function.
Then there exists an extreme point $x_0\in \operatorname{ext}K$ such that
$$
f(x_0)=\max_{x\in K} f(x)
$$
holds.
\end{proposition}
We use the Bauer maximum principle to establish the existence of a maximizer that is an extreme point and then show that the extreme-point condition forces this maximizer to be a characteristic function.
\begin{proof}[Proof of Proposition~\ref{prop:characteristic_mesa}]
Minimizing the limiting energy \(E_\infty\) is equivalent to maximizing the nonlocal attraction energy $Q[u]$ over the capacity-constrained set
\(\mathcal A_M\).

We verify that $\mathcal{A}_M$ and the nonlocal interaction energy $u \mapsto \int_{\Omega} u (K\ast u)$ defined on it satisfy the hypotheses of the Bauer maximum principle.
It is known that $L^\infty(\Omega)$, endowed with the weak-$\ast$ topology, is a locally convex Hausdorff space \cite[Comments on Chapter 3]{brezis}.
Moreover, its subset \(\mathcal A_M\) is a convex set that is compact in the weak-\(*\) topology of $L^\infty(\Omega)$.
This follows immediately from the Banach--Alaoglu theorem.
Furthermore, the nonlocal attraction energy is a convex weak-\(*\) continuous functional.
Indeed, by assumption,
$
    \int_\Omega u(K*u)\,dx\ge0
    \quad
    (u\in L^2(\Omega))
$
and hence, for \(0\le\theta\le1\),
\[
\begin{aligned}
& (1-\theta)\int_\Omega u(K*u)\,dx
  +\theta\int_\Omega v(K*v)\,dx \\
&\quad
  -\int_\Omega \bigl((1-\theta)u+\theta v\bigr)
  K*\bigl((1-\theta)u+\theta v\bigr)\,dx \\
&=
\theta(1-\theta)
\int_\Omega (u-v)\bigl(K*(u-v)\bigr)\,dx
\ge0.
\end{aligned}
\]
Moreover, since $K \in L^1(\Omega)$, \eqref{eq:interaction_weak_continuity} shows that this nonlocal attraction energy is continuous with respect to weak-$\ast$ convergence in $L^\infty(\Omega)$. Note here that, because $L^1(\Omega)$ is separable, the weak-$\ast$ topology on the subset $\mathcal A_M$ of $L^\infty(\Omega)$ is metrizable, and therefore sequential upper semicontinuity is equivalent to upper semicontinuity.
The preceding argument and the Bauer maximum principle show that the nonlocal attraction energy attains its maximum at an extreme point.
Consequently, a minimizer of \(E_\infty\) can be chosen to be an extreme point of \(\mathcal A_M\).

It remains to characterize the extreme points of $\mathcal A_M$. We first show that if $u\in \mathcal A_M$ satisfies
$0<u<1$ on a set of positive measure, then $u$ is not an extreme point of
$\mathcal A_M$.
Indeed, suppose that
$
    F:=\{x\in\Omega:0<u(x)<1\}
$
has positive measure. Since
$
    F=\bigcup_{n=2}^{\infty}
    \left\{x\in\Omega:\frac1n<u(x)<1-\frac1n\right\}
$
countable subadditivity implies that, for some $n_{0}\ge 2$,
\[
    B:=\left\{x\in\Omega:\frac{1}{n_{0}}<u(x)<1-\frac{1}{n_{0}}\right\}
\]
has positive measure.

Choose a function $\xi_B\in L^\infty(\Omega)$ supported in $B$ such that
\[
    \int_\Omega \xi_B\,dx=0,\qquad
    \|\xi_B\|_{L^\infty}\le 1
\]
holds. Indeed, since Lebesgue measure is atomless, \cite[Corollary 1.12.10]{bogachev} yields a measurable subset $B_1$ of $B$ such that $|B_1|=\frac{|B|}{2}$.
Thus, setting $B_2:=B\setminus B_1$, it suffices to define $\xi_B$ by
\[
    \xi_B:=\mathbf{1}_{B_1}-\mathbf{1}_{B_2}.
\]

Using this $\xi_B$, define
\[
    u_+ := u+\frac{1}{n_{0}} \xi_B,\qquad
    u_- := u-\frac{1}{n_{0}} \xi_B.
\]
Since $\xi_B$ vanishes outside $B$, while
$\frac{1}{n_{0}}<u<1-\frac{1}{n_{0}}$ and $|\xi_B|\le 1$ on $B$, it follows that
$
    0\le u_\pm\le 1
$
holds. Moreover, since $\int_\Omega \xi_B\,dx=0$,
$
    \int_\Omega u_\pm\,dx=\int_\Omega u\,dx=M
$
and hence $u_\pm\in \mathcal A_M$. Furthermore, since $\xi_B\not\equiv 0$,
$u_+\ne u_-$, and
$
    u=\frac12(u_++u_-)
$
holds. Therefore, $u$ is not an extreme point of $A_M$.

Thus, if $u \in \mathcal{A}_M$ is an extreme point, then $|\{0<u<1\}|=0$, and hence there exists a measurable set $A$ such that $u=\mathbf{1}_{A}$.
The capacity constraint further implies that $|A|=M$.

We conclude that \(E_\infty\) admits a minimizer of the form
$
    u_\infty=\mathbf 1_A,
    \quad
    |A|=M.
$
\end{proof}

\subsection{Proof of the threshold condition}
\label{app_subsec:threshold_condition_proof}

We prove Proposition~\ref{prop:mesa_threshold_condition} using the maximization property of the nonlocal attraction energy.
\begin{proof}[Proof of Proposition~\ref{prop:mesa_threshold_condition}]
Let \(u_\infty\in\mathcal A_M\) be a minimizer of \(E_\infty\).
Then \(u_\infty\) maximizes the nonlocal attraction energy $Q[u]$
over \(\mathcal A_M\).

Set
$
    \psi:=K*u_\infty=T_{K}u_{\infty}
$.
We show that \(u_\infty\) maximizes the linear functional
\[
    v\longmapsto \int_\Omega v\psi\,dx
\]
over \(\mathcal A_M\).
For any \(v\in\mathcal A_M\) and \(0\le t\le1\),
$
    u^t:=(1-t)u_\infty+tv
$
also belongs to \(\mathcal A_M\).
By the maximality of \(u_\infty\),
\[
    \int_\Omega u^t(T_Ku_t)\,dx
    \le
    \int_\Omega u_\infty(T_Ku_\infty)\,dx.
\]
Setting \(w:=v-u_\infty\), we directly compute the difference between the two sides:
\[
\begin{aligned}
0
&\ge
\int_\Omega u^t(T_Ku_t)\,dx
-
\int_\Omega u_\infty(T_Ku_\infty)\,dx  \\
&=
\int_\Omega (u_\infty+tw)T_K(u_\infty+tw)\,dx
-
\int_\Omega u_\infty(T_Ku_\infty)\,dx  \\
&=
t\int_\Omega w(T_Ku_\infty)\,dx
+
t\int_\Omega u_\infty(T_Kw)\,dx
+
t^2\int_\Omega w(T_Kw)\,dx .
\end{aligned}
\]
Since \(K\) is even, \(T_K\) is a symmetric operator, and hence
$
    \int_\Omega u_\infty(T_Kw)\,dx
    =
    \int_\Omega w(T_Ku_\infty)\,dx.
$
Therefore,
\[
0
\ge
2t\int_\Omega w(T_Ku_\infty)\,dx
+
t^2\int_\Omega w(T_Kw)\,dx .
\]
Using \(\psi=T_Ku_\infty\) and \(w=v-u_\infty\), we obtain
\[
0
\ge
2t\int_\Omega (v-u_\infty)\psi\,dx
+
t^2\int_\Omega (v-u_\infty)T_K(v-u_\infty)\,dx .
\]
Dividing by \(t>0\) and then letting \(t \to 0\), we obtain
\begin{equation}
    \int_\Omega v\psi\,dx
    \le
    \int_\Omega u_\infty\psi\,dx
    \qquad
    (v\in\mathcal A_M)
    \label{eq:threshold_linear_max}
\end{equation}

We now use \eqref{eq:threshold_linear_max} to show that
$
    \operatorname*{ess\,sup}_{\{0\le u_\infty<1\}}\psi
    \le
    \operatorname*{ess\,inf}_{\{0<u_\infty\le 1\}}\psi.
$
We argue by contradiction.
Suppose that
$
    \operatorname*{ess\,sup}_{\{0 \le u_\infty<1\}}\psi
    >
    \operatorname*{ess\,inf}_{\{0<u_\infty \le 1\}}\psi.
$
Then there exist numbers \(\alpha_1<\alpha_2\) and disjoint measurable sets \(E,F\) of positive measure such that
\[
    E\subset \{0 \le u_\infty<1\},
    \qquad
    F\subset \{0<u_\infty\le 1\},
\]
\[
    \psi>\alpha_2
    \quad \text{on } E,
    \qquad
    \psi<\alpha_1
    \quad \text{on } F.
\]
By passing to subsets if necessary, we may further assume that, for some \(\delta>0\),
$
    E\subset \{0 \le u_\infty\le1-\delta\}
$,
$
    F\subset \{\delta \le u_\infty \le 1\}
$
and \(|E|=|F|>0\).
Using these sets $E,F$, define
$
    v
    =
    u_\infty
    +
    \delta(\mathbf 1_E-\mathbf 1_F).
$
Then
$
    0\le v\le1,
    \quad
    \int_\Omega v\,dx=M
$
and hence \(v\in\mathcal A_M\). On the other hand,
\[
    \int_\Omega (v-u_\infty)\psi\,dx
    =
    \delta
    \left(
        \int_E\psi\,dx
        -
        \int_F\psi\,dx
    \right)
    >
    \delta(\alpha_2-\alpha_1)|E|
    >
    0.
\]
This contradicts \eqref{eq:threshold_linear_max}.

Therefore,
$
    \operatorname*{ess\,sup}_{\{0\le u_\infty<1\}}\psi
    \le
    \operatorname*{ess\,inf}_{\{0<u_\infty\le 1\}}\psi.
$
Choosing a constant \(\lambda\) between these two values, we obtain
\[
    \psi\le\lambda
    \quad
    \text{a.e. on } \{0 \le u_\infty<1\},
\]
\[
    \psi\ge\lambda
    \quad
    \text{a.e. on } \{0<u_\infty \le 1\}.
\]
In particular,
\[
    K*u_\infty\ge\lambda
    \quad
    \text{a.e. on } \{u_\infty=1\},
\]
\[
    K*u_\infty\le\lambda
    \quad
    \text{a.e. on } \{u_\infty=0\}.
\]
Moreover, in the region where \(0<u_\infty<1\), both inequalities hold simultaneously, and therefore
\[
    K*u_\infty=\lambda.
\]

This proves Proposition~\ref{prop:mesa_threshold_condition}.
\end{proof}

\end{document}